\documentclass[a4paper, intlimits, 10pt]{amsart}

\usepackage[T1]{fontenc}
\usepackage{float}

\usepackage{times}
\usepackage[english]{babel}
\usepackage[T1]{fontenc}
\usepackage{enumitem}
\usepackage{graphicx}
\usepackage{tikzsymbols}
\usepackage[font=small,labelfont=bf]{caption}

\usepackage{subcaption}
\usepackage{booktabs}
\usepackage{multirow}
\usepackage{tabularx}
\usepackage{color}
\usepackage{eurosym}
\usepackage[numbers,sort&compress]{natbib}

\usepackage{xcolor}
\usepackage{comment}

\definecolor{darkgreen}{RGB}{0,170,0}

\usepackage{tikz}
\usetikzlibrary{
  fit,
  positioning,
  shapes.geometric,
  arrows,
}

\usepackage{amsthm, amssymb, amsfonts}
\usepackage{amsmath}

\DeclareMathOperator*{\argmin}{arg\,min}

\usepackage{accents}

\newtheorem{theorem}{Theorem}[section]

\newtheorem{proposition}[theorem]{Proposition}
\newtheorem{definition}[theorem]{Definition}

\newtheorem{remark}[theorem]{Remark}
\numberwithin{equation}{section}

\begin{document}

\tikzset{
  arrowss/.style = {
    thick,->,>=stealth,
    line width=2pt,
    color = black,
  }
}

\title{On Detecting Spoofing Strategies in High Frequency Trading}

\date{\today}

\author{Xuan Tao}
\address{Shanghai Jiao Tong University, Shanghai, China}
\email{taoxuan@sjtu.edu.cn}

\author{Andrew Day}
\address{Western University, Canada}
\email{aday46@uwo.ca}

\author{Lan Ling}
\address{Shanghai Jiao Tong University (now at Ping An Technology), Shanghai, China}
\email{lan.ling@sjtu.edu.cn}
\email{linglan987@pingan.com.cn}

\author{Samuel Drapeau}
\thanks{
We thank the Fields Institute for the organization of the many ``Fields-China Joint Industrial Problem Solving Workshop'' from which this problem stems.
We also thank TMX and in particular the TMX Analytics Team for supporting this project with profound datasets, high performing computing facilities as well as precious market insights in high frequency trading.
The financial supports from the National Science Foundation of China, Grants Numbers: 11971310 and 11671257; as well as from Shanghai Jiao Tong University, Grant number AF0710020; are gratefully acknowledged.
Andrew Day gratefully acknowledges the funding provided by Mitacs, the Ontario Graduate Scholarship (OGS), and the Natural Sciences and Engineering Research Council of Canada postgraduate scholarship (NSERC PGS)
}
\address{Shanghai Jiao Tong University, Shanghai, China}
\email{sdrapeau@saif.sjtu.edu.cn}
\urladdr{http://www.samuel-drapeau.info}

\begin{abstract}
    Spoofing is an illegal act of artificially modifying the supply to drive temporarily prices in a given direction for profit.
    In practice, detection of such an act is challenging due to the complexity of modern electronic platforms and the high frequency at which orders are channeled.
    We present a micro-structural study of spoofing in a simple static setting.
    A multilevel imbalance which influences the resulting price movement is introduced upon which we describe the optimization strategy of a potential spoofer.
    We provide conditions under which a market is more likely to admit spoofing behavior as a function of the characteristics of the market.
    We describe the optimal spoofing strategy after optimization which allows us to quantify the resulting impact on the imbalance after spoofing.
    Based on these results we calibrate the model to real Level 2 datasets from TMX, and provide some monitoring procedures based on the Wasserstein distance to detect spoofing strategies in real time.

    \vspace{5pt}

\end{abstract}

\maketitle
\section{Introduction}

The act of spoofing is a specific trading activity that aims at artificially modifying the supply on the market, without intend to trade, to move it away from its equilibrium.
One might profit from the resulting short term price movement by canceling the previous supply while the market comes back to its equilibrium.
Such a strategy implies that the spoofer should be able to act anonymously, fast and in a market where all the other agents react to offer and demand.
In this regard, with the recent rise of centrally cleared venue and high frequency algo-trading, the ground for the existence of spoofing schemes is rising, see \citet{shorter2015}.\footnote{In 2010, trader Navinder Singh Sarao was accused of exacerbating a flash crash by placing thousands of E-mini S\&P 500 stock index futures contract orders in one day and changed or moved those orders more than 20 million times before they were cancelled.}
In a competitive market where many potential spoofers are present, spoofing behavior might cancel out, but most regulations consider it as illegal.
For instance, the 2010 Dodd-Frank Act prohibits spoofing -- defined as activity of bidding or offering with the intent to cancel before execution -- that can be prosecuted as ``a felony punishable by up to \$1 million in penalties and up to ten years in prison for each spoofing count''.\footnote{
In 2019, the high frequency company Tower Research Capital agreed to pay a fine of about \$60 million over spoofing allegations.
In 2020, JP Morgan settles spoofing lawsuit alleging fraud for about \$920 million.
}

Yet, for several reasons, detecting and prosecuting spoofing behavior is a challenging problem.
First is the sheer amount of data produced from high frequency trading across many financial products and venues.
Second, it is usually impossible to trace in real time who is behind every trade.
From a CCP viewpoint, they mainly have access to the broker ID through which the trade has been channeled resulting only in aggregated informations.
Furthermore, a potential spoofer might post those trades through different venues and brokers.
Third, aside from a loose definition, it is unclear how a spoofing strategy differs quantitatively from other strategies and what is the resulting impact: Thus, the complexity of quantifying and discriminating spoofing strategies from legitimate ones.
Finally, due to the lack of information from the few regulatory cases, the problem amount to an unsupervised classification problem.
Based on the above points, it seems difficult to provide an efficient way to monitor the market for spoofing behavior.\footnote{Aside from obvious cases or exogenous approaches as for insider prosecution.}
However, a potential spoofer is also confronted to the constraints of modern electronic trading platforms.
Indeed, from the basic spoofing description, the spoofer has to act rapidly in a complex and high frequency environment.
Therefore, it must rely on fast -- henceforth simple -- algorithmic strategies which, due to the complexity of the dynamic structure of a limit order book, is based on aggregated signal.

Along these lines, and in view of this unsupervised classification task, we intentionally address a quantitative analysis of spoofing in a simple setup.
As a basis for this study, let us consider a simple example.
We suppose that in the next period the limit order book shifts up by one unit with a probability $\bar{p}$ and down by one unit otherwise.
From the perspective of an agent whose objective is to purchase two shares, it faces the following three idealized situations.
\begin{enumerate}[label = \arabic* - , fullwidth]
    \item Immediately post a buy market order for a total cost of
        \begin{equation*}
            \hat{C} = 10+11 =21
        \end{equation*}
        \begin{center}    
            \begin{tikzpicture}[node distance=10em]
                \node (lmo1) {\includegraphics[width=0.4\textwidth]{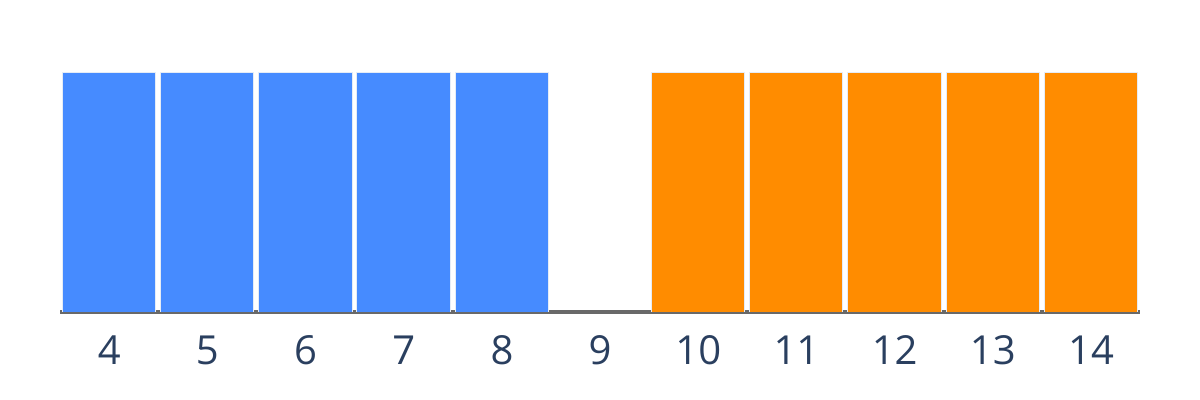}};
                \node (lmo2) [right of = lmo1, xshift =0.2\textwidth] {\includegraphics[width=0.4\textwidth]{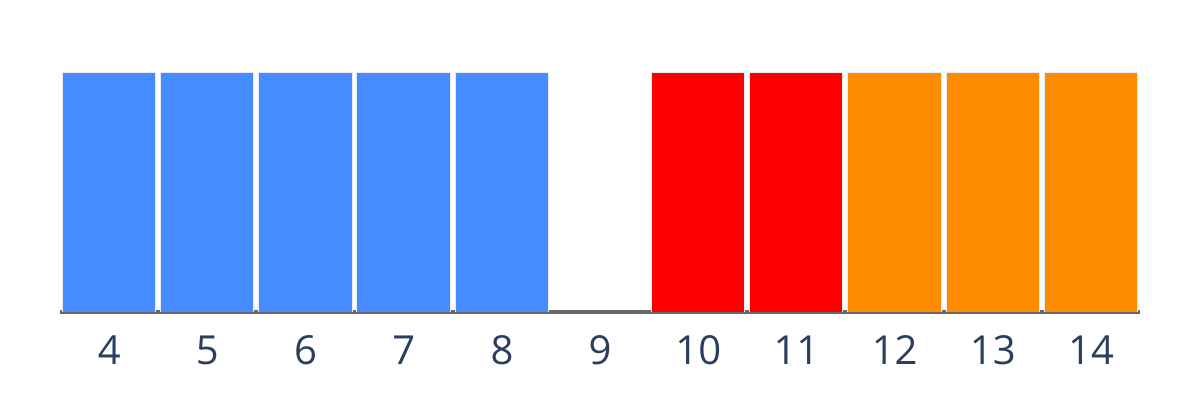}};
                \draw [arrowss] (lmo1)--(lmo2);
            \end{tikzpicture}
        \end{center}
    \item Delay the buy market order for one period resulting in a total average cost of
        \begin{equation*}
            \tilde{C} = (1-\bar{p})\left( 9+10 \right)+\bar{p} \left( 11+12 \right) = \left( 1-\bar{p} \right)19 + \bar{p}23
        \end{equation*}
        which is smaller than $\hat{C}$ if and only if $\bar{p}<1/2$, hence a bearish market.
        
        \begin{center}
            \begin{tikzpicture}[node distance=10em]
              \node (lmo1) {\includegraphics[width=0.4\textwidth]{./figs/lmo.png}};
              \node (lmo2) [right of = lmo1, xshift = 0.2\textwidth, yshift = 0.1\textwidth] {\includegraphics[width=0.4\textwidth]{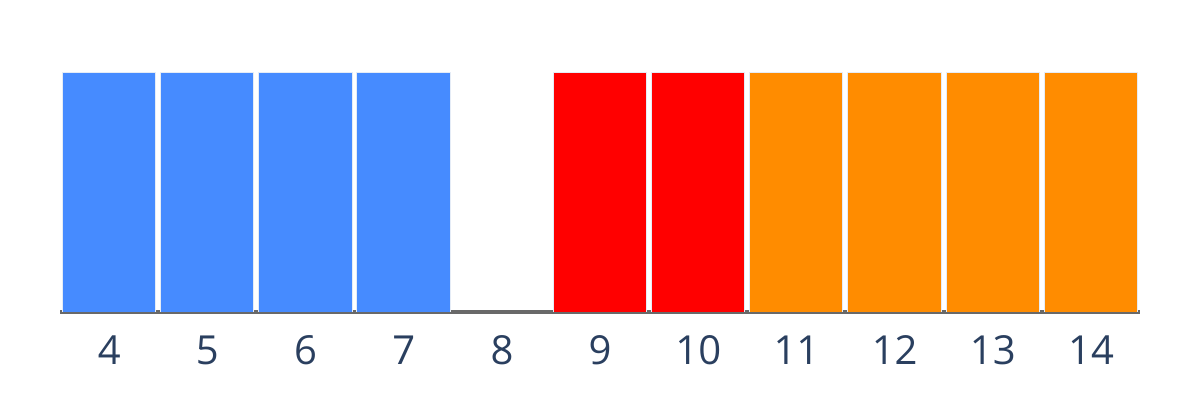}};
              \node (lmo3) [right of = lmo1, xshift = 0.2\textwidth, yshift = -0.1\textwidth] {\includegraphics[width=0.4\textwidth]{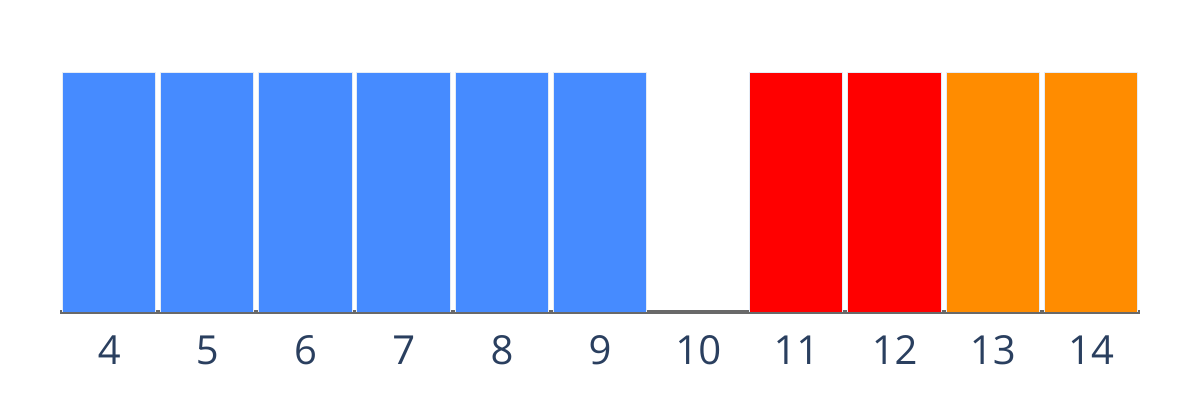}};
            
              \draw [arrowss] (lmo1)-- node[color = black,yshift = 0.1cm,above]{$ 1- \bar{p}$} (lmo2);
              \draw [arrowss](lmo1)-- node[color = black,yshift = 0.1cm,above]{$\bar{p}$} (lmo3);
            \end{tikzpicture}
        \end{center}
    \item Delay the buy market order and post a sell limit order of one share at a distance of one unit from the best ask price to artificially modify the offer and demand resulting in a temporary more bearish state $p<\bar{p}$.\footnote{We suppose that during this time period, if the market moves through limit order posting/canceling, the agent keeps its sell limit order one unit away from the best ask price by rapidly canceling and posting again.}
        Doing so, with a probability $q$, its sell limit order is executed through an incoming market order walking the limit order book beyond one unit.
        For this executed sell order, the agent receives an average price of $q ( (1-p) 10 + p12 )$ while its inventory increases in average to $2 + q$.
        The cost of buying back this increased inventory minus the gain from selling its limit order results in an average net cost of 
        \begin{align*}
            C & = (1-p)\left( 9 + 10 + q 11 \right) + p\left( 11+12+q 13 \right) - q \left( (1-p) 10 + p12 \right) \\
              & = (1-p)19+p23 + q
        \end{align*}

        \begin{center}
            \begin{tikzpicture}[node distance=10em]
                \node (lmo1) {\includegraphics[width=0.4\textwidth]{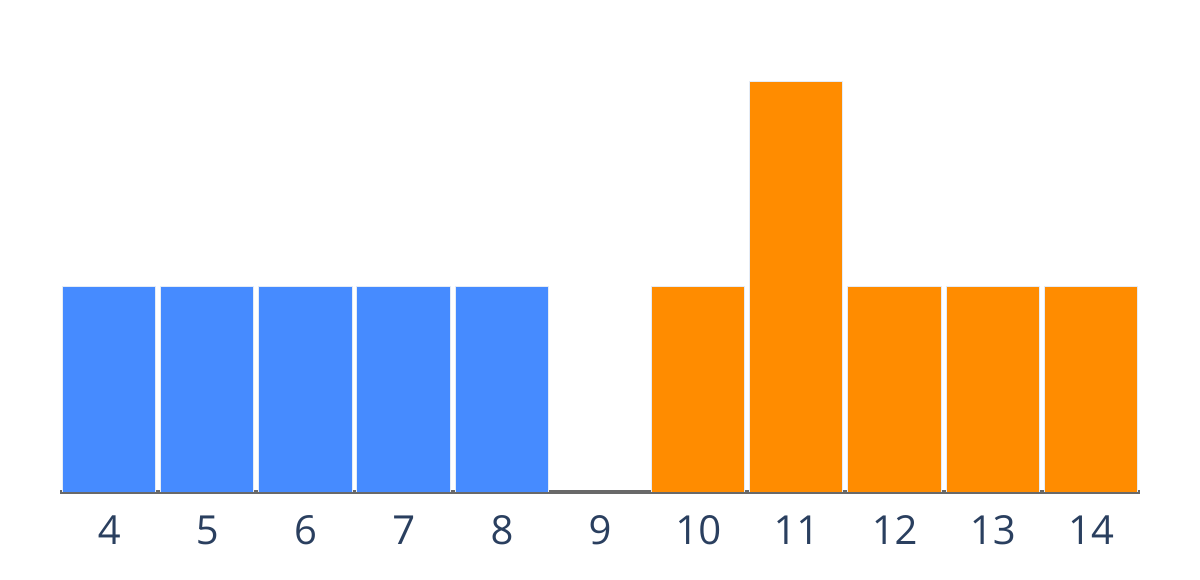}};
                \node (lmo2) [right of = lmo1, xshift = 0.2\textwidth, yshift = 0.1\textwidth] {\includegraphics[width=0.4\textwidth]{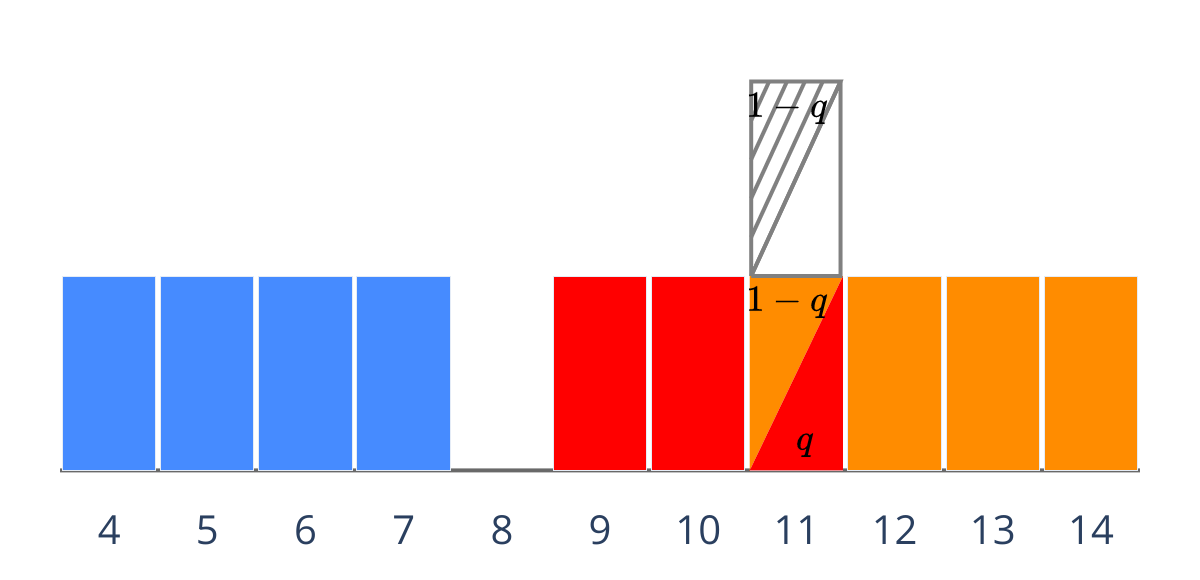}};
                \node (lmo3) [right of = lmo1, xshift = 0.2\textwidth, yshift = -0.1\textwidth] {\includegraphics[width=0.4\textwidth]{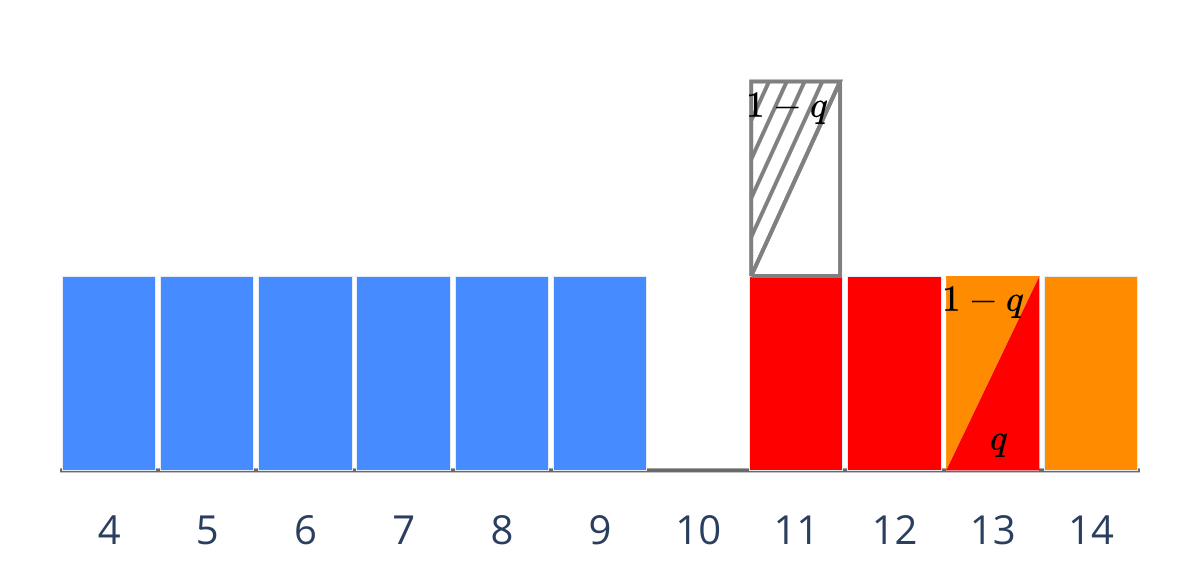}};
                \draw [arrowss] (lmo1)-- node[color = black,yshift = 0.1cm,above]{$ 1- p$} (lmo2);
                \draw [arrowss] (lmo1)-- node[color = black,yshift = 0.1cm,above]{$p$} (lmo3);
            \end{tikzpicture}
        \end{center}
\end{enumerate}
From this simple situation, a profitable spoofing situation depends on the value of $\bar{p}$, how bearish $p<\bar{p}$ the market reacts to an increase of one share at a distance of one unit from the best ask price, as well as the probability $q$ of being adversarially executed through a market order during this time.
The bottom line is this paper is to study the interplay between these different factors and the resulting impact in view of detection procedure, in particular the impact on the price movement as a function of the spoofing size as well as the depth in the limit order book.
Based on this theoretical approach, we present some approaches to track spoofing behavior and calibrate those to real market data.

We model the impact of the offer and demand on the price through the volume imbalance often taken as the ratio of the volume of the best bid divided by the total volume on the best bid and ask.
Since the spoofer never intends to have their orders executed, spoofing is more likely to happen beyond the top of the limit order book.
Indeed, the probability of getting executed is too high, resulting in a negative payoff.
To take this into account we weight the impact on the imbalance in terms of depth as follows
\begin{equation*}
    \bar{\imath} = \frac{ \sum \bar{v}^-_k w_k}{\sum \left( \bar{v}^-_k + \bar{v}^+_k \right)w_k}
\end{equation*}
where $\bar{v}^{\pm}_k$ represents the volume on the bid/ask $k$ units away from the best bid/ask and $w_k$ is the relative impact on the imbalance at level $k$.
If the agent posts a sell limit order $v$ on the ask side at tick level $k_0$, the imbalance moves to
\begin{equation*}
    \imath(v) = \frac{ \sum \bar{v}^-_k w_k}{\sum \left( \bar{v}^-_k + \bar{v}^+_k \right)w_k + w_{k_0}v} \leq \bar{\imath}
\end{equation*}
If we denote by $dp_n$ the probability of a price deviation of $n$ units, the dependence on the imbalance $\imath$ is given as follows
\begin{equation*}
    dp_n = \imath dp^+_n + (1-\imath) dp^-_n, \quad n = \ldots, -2, -1, 0, 1, 2, \ldots
\end{equation*}
where $dp^\pm_n$ represents the price deviation when the imbalance is at its extreme.
When the imbalance $\imath$ is low/high -- the offer/demand dominates -- the price distribution is biased downwards/upwards through $dp^\pm$.
The agent can influence the price distribution in a non linear way through the volume it posts:
\begin{equation*}
    v \longmapsto dp(v) := \imath(v) dp^+ + \left( 1-\imath(v) \right) dp^-
\end{equation*}
Given the probability $dq$ of a sell limit order hitting the limit order book up to a given level the resulting net average cost of spoofing turns out to be
\begin{equation*}
    C(v) = pH + \underbrace{(1-Q) G(H)}_{\text{Cost for optimal situation}}  + \underbrace{ H \mu^+ \left( 2\imath(v)-1 \right)}_{\text{Spoofing impact}} + \underbrace{Q G(H+v)+  v \nu}_{\text{Cost for being caught wrong way}} 
\end{equation*}
where $H$ is the initial objective of shares to acquire, $G$ is the liquidity costs from walking through the limit order book, $\mu^+>0$ is the mean of $dp^+$, $Q$ (resp $\nu$) is the probability (resp mean beyond $k_0$) of being executed beyond $k_0$.
From this expression, there is a competitive aspect between the risk of being caught on the wrong side and the fact of pushing $\imath(v)$ way beyond $1/2$ to get $H\mu^+(2\imath(v) -1)\leq 0$.

In a theoretical part we first provide conditions for the limit order book to admit spoofing manipulation\footnote{In other words, better than immediate or delayed market order.} as a function of the initial imbalance $\bar{\imath}$, local sensitivity of the imbalance $w$, overall price impact $\mu^+$, liquidity cost $G$, as well as the objective $H$.
In short, this model confirms several intuitive facts when spoofing is more likely to occur
\begin{itemize}
    \item if the probability $Q$ for the spoofing order to be executed is small;
    \item if the local sensitivity $w_{k_0}$ or the overall price impact $\mu^+$ is large;
    \item if the amount of share $H$ to buy is large with respect to the depth of the limit order book;
    \item if the initial market imbalance is close to $1/2$, that is, the market is equally balanced between offer and demand.
    \item away from the top of the limit order book;
\end{itemize}
As for the depth of the limit order book -- how liquid the market is -- the results are inconclusive.
For this to be taken into account, one should model how the above mentioned parameters depend on the liquidity of the limit order book.
The subsequent empirical study shows that it is the case, but we can not derive conclusions from this model as in \citet{shorter2015} where illiquid markets seems more prone to spoofing.
We then address the impact of spoofing on the resulting imbalance $\imath_{spoof} = \imath(v_{spoof})$ and discuss its dependence as a function of the aforementioned parameters.
We characterize and discuss the deviation for the imbalance as a function of the different parameters.
In particular as a function of the depth where the spoofing order is posted.
We finally address the situation of a market maker using spoofing strategies for a positive round trip payoffs.

Based on this study, we can theoretically discriminate a spoofed imbalance $\imath_{spoof}$ from the legitimate one $\bar{\imath}$.
Yet, from an outsider perspective, this is a hidden value since only the spoofer is aware of the actual imbalance.
The main idea for the detection is to observe that a successful spoofing strategy requires the execution of a market order.
We therefore compare the imbalance $\imath_-(t)$ before a market order at time $t$ with the imbalance $\imath_+(t)$ after this market order.
If the market order is legitimate, the behavior of these two imbalances should follow some steady state distribution $(\bar{\imath}_-, \bar{\imath}_+)$.
On the other hand, if the market order is the result of a spoofing strategy, the imbalance before the market order should be of the form
\begin{equation*}
    \imath_{spoof} \approx \frac{b}{b+a + wv_{spoof}}<\frac{b}{a+b}=\bar{\imath}_-
\end{equation*}
while returning to its steady state $\bar{\imath}_+$ as soon as the spoofed volumes are canceled.
Hence, a quantification approach is to compare the distance from the instant imbalance $\imath_-(t)$ before a market order at time $t$ with, one the one hand, the long term legitimate one $\bar{\imath}_-$, and with, on the other hand, the theoretical spoofed one $\imath_{spoof}$.
This measure is done according to the current market state situation $\imath_{+}(t)$.
In other terms we measure and compare
\begin{equation*}
    \underbrace{d\left(\imath_-(t), \bar{\imath}_-|\imath_+(t)\right)}_{\substack{
            \text{distance of instant imbalance }\imath_-(t)\\
            \text{before market order to legitimate imbalance }\bar{\imath}_-\\
    \text{given current market conditions }\imath_+(t)}}
    \quad \text{and}\quad 
    \underbrace{d\left(\imath_-(t), \imath_{spoof} |\imath_+(t)\right)}_{\substack{
            \text{distance of instant imbalance }\imath_-(t)\\
            \text{before market order to spoofing imbalance }\imath_{spoof}\\
    \text{given current market conditions }\imath_+(t)}}
\end{equation*}
For the distance, we adopt the non parametric Wasserstein distance.
The technical design, in particular in terms of conditioning, the calibration with market data, implementation as well as the reason for such an approach are explained and illustrated for several stocks from TMX.

Before addressing the relevant literature, let us expose the shortcomings and modeling choices of this approach.
The micro-structure dynamic of the limit order book at high frequency is complex.
To excerpt some key impacts of spoofing behavior we deliberately focus on a static situation where the dynamic of the market is ignored.\footnote{Since we consider the limit order book beyond its top, a dynamic version of the present approach would result into a fairly complex and high dimensional dynamic programing problem.}
For instance, we do not consider situations where compound spoofing behavior happens.
We furthermore assume that there exists a single potential spoofer and that the market is infinitely reactive in the sense that it comes back to its steady state driven by the imbalance.
There is no implicit game where the market acknowledges the existence of a potential spoofer, such as in \citet{kyle1985} for instance.
Finally, we assume that spoofing behavior is happening in a single market, while it has been documented and studied by \citet{kervel2015} that fast traders take advantage of multiple venues to post market orders in one while cancelling their limit orders in other venues.
However, the present approach could also take into account an average imbalance over several venues.
There is also no competitive game between two or more spoofers.
Also, even though we shortly address the situation of a round trip for a market maker and the resulting optimal spoofing behavior, we take the viewpoint of a market taker willing to purchase/sell a given amount of shares.
The overall goal being to understand the mechanism of spoofing in its most simple nature, quantify the resulting impact and derive potential detection procedures.
Refinement of this approach, other take on, as well as more adequate quantification procedures are topics of further studies.

\subsection{Literature review}
There exists a solid stream of research showing that even rational speculative activities might destabilize prices, have an adverse effect on market efficiency or eventually lead to different forms of arbitrage;
From market speculation based on various form of information asymmetries, for instance \citet{kreps1986}, \citet{allen1992} or \citet{jarrow1992}, to price manipulation in limit order books using different market impact assumptions and trading strategies, as studied in \citet{alfonsi2007, alfonsi2010, gatheral2009, gatheral2013}.
The specific case of spoofing behavior has not yet been the subject of much theoretical study.

Although many high frequency trading strategies are legitimate, \citet{shorter2015} point out that high frequency trading firms may engage in potentially manipulative strategies involving the usage of quote cancellations.
\citet{lee2013microstructure} empirically study the change in spoofing behavior following a modification in volume disclosure rules on the Korean Exchange (KRX) at the start of 2002.
Up to the end of 2001 the KRX disclosed the total volume of shares on both sides of the book and also the volumes at each tick up to 5 ticks from the best ask/bid.
At the start of 2002 the KRX stopped disclosing the total volume on both sides in an effort to stop spoofing, but increased the disclosed volumes at each tick from the first 5 to the first 10 ticks from the best ask/bid.
They show that spoofing is profitable and spoofers tend to prefer stocks with higher return volatility, lower market capitalization, lower price and lower managerial transparency.
This study suggests the importance of the depth of book on spoofing strategies and potential price manipulation being carried out through a form of ``volume imbalance''.
\citet{wang2015strategic} followed a similar methodology in empirically studying spoofing on Taiwan's index futures market.
They found consistent results on the impact of spoofing on the market, but without the novel testing ground on changes in the disclosure of volumes deeper in the limit order book.

Some other studies to detect price manipulations are mainly based on learning algorithms.
Among other, \citet{cao2014, cao2015} based on the definition of spoofing in \citep{lipton2013} use K-nearest neighbour, one class support vector machine and adaptive hidden markov models to classify the data.
\citet{enrique2016} characterize spoofing and pinging as full and partial observability of Markov decision processes.
Under a reinforcement learning framework, they find that in order to maximise the investment growth, a trader will always employ spoofing or pinging orders except when market adds extra transaction costs or fines.
In contrast to these empirical studies, our approach focuses on the micro economic features of spoofing behavior, in particular using our main stylized signal, the imbalance, which measures the difference between bid and ask side. 

Concerning the impact of the imbalance on direction of the price movement:
\citet{lipton2013} use the definition on the top of the book for the imbalance and study the impact on the trade arrival dynamic and resulting price movement.
They fit a stochastic model for this behavior on real market data.
\citet{cartea2018enhancing} employ volume imbalance as a signal to improve profits on the liquidation of a collection of shares in a dynamic high-frequency trading environment.
\citet{gould2015} fit logistic regressions between the imbalance and the direction of the subsequent mid-price movement for each of 10 liquid stocks on Nasdaq, and illustrates the existence of a statistically significant relationship.
\citet{xu2019} compute the imbalance at multi position in the limit order book and fit a linear relationship between this imbalance and the mid-price change.
They find that the goodness-of-fit is considerably stronger for large-tick stocks than it is for small-tick stocks.
The impact of order imbalance on prices has also been studied by \citet{cont2014price} and \citet{bechler2017order}, for example.
\citeauthor{bechler2017order} also found that including characteristics of deeper parts of the book may be necessary for forecasting price impact. However, due to the nature of their dataset, they were only able to look at an aggregated form of the depth of book while we are able to use the exact volumes at all depths in the book.
\citet{sirignano2019deep} also used the book volumes beyond the touch to model price movements in a deep learning setting. Further suggesting the impact of depth of book on predicting future price movements.

Finally, the closest work to the present one in terms of quantitative analysis of spoofing behavior in relationship with imbalance is from \citet{cartea2020}.
They adopt a dynamic approach where the trader influences the imbalance to derive the optimal strategy.
They calibrate their model to market data and provide trading trajectories for the spoofer showing that spoofing considerably increases the revenues from liquidating a position.
While being in a dynamic setting, in contrast to the present study, everything happens at the top of the limit order book for the imbalance to be manipulated.
Furthermore, we do not focus here on the resulting gains from the spoofer, be rather on the impact on the imbalance from spoofer as for detection purposes from a regulatory viewpoint.

\subsection{Organization of the paper}
The first Section introduces the model, the imbalance and the dependence of the price movement on that imbalance.
The second Section presents the spoofing strategies, addresses the theoretical conditions for spoofing behavior to happen and provide the resulting imbalance after spoofing together with numerical illustrations.
The third Section is dedicated to the calibration procedure of the model on real Level 2 market data from TMX.
The last Section discusses and introduces a quantitative approach to track spoofing behavior in real time illustrated on real datasets.
Proofs, treatment of the round trip situation, calibration details and conditional distance specification are content of the Appendix.

\section{Limit Order Book, Liquidity Costs and Imbalance}
The ask price is denoted by $p$ and the limit order book on the ask side by $\bar{v} = (\bar{v}_{0},\bar{v}_1, \ldots)$, that is, $\bar{v}_0$ is the volume posted at ask price $p$, $\bar{v}_1$ the volume posted at $p+\delta $, etc. where $\delta$ is the tick size.
We denote by $\bar{v}^- = (\bar{v}_0^-, \bar{v}^-_{1}, \ldots)$ the limit order book on the bid side.\footnote{That is $\bar{v}^-_0$ is the volume posted at bid price $p^-<p$, $\bar{v}^-_1$ the volume posted at $p^- - \delta$, etc.}

Given a limit order book inventory $\bar{v}$ on the ask side, we define for an amount $H\geq 0$ of shares the function
\begin{equation*}
    F(H)  := \inf \left\{ x \in \mathbb{N}_0\colon \sum_{k=0}^{x} \bar{v}_k \geq H \right\}
\end{equation*}
which represents how many positive price tick deviation an order of size $H$ generates.
Given an amount of shares $H$ to buy, a bid price $p$ and an ask limit order $\bar{v}$, the resulting costs of the market order is
\begin{equation*}
    pH + \sum_{k=0}^{F(H)} k \delta \bar{v}_k - \delta F\left( H \right)\left( \sum_{k=0}^{F(H)} \bar{v}_k - H \right) = pH +\delta G(H)
\end{equation*}
The term $G$ on the right hand side represents the liquidity costs depending only on $\bar{v}$.
\begin{remark}
    Throughout the theoretical part of this work we assume that the limit order book is blocked shaped with an amount $a>0$ of shares at each price level of the ask side.
    We then get the continuous approximation 
    \begin{equation*}
        F(H) \approx \frac{H}{a} \quad \text{and} \quad G(H)\approx \frac{H^2}{2 a}
    \end{equation*}
\end{remark}
As for the imbalance of the limit order book, measure of the difference between offer and demand, we proceed as follows.
Let $w_0$, $w_1$, \ldots with $\sum w_k =1$ and $w_k\geq 0$, a weight for each tick level $k$, and a limit order book inventory $\bar{v}^-$, $\bar{v}$ on the bid and ask respectively, we denote by
\begin{align*}
    \bar{B} & = \sum w_k \bar{v}^-_k = \langle w, \bar{v}^-\rangle &
    \bar{A} & = \sum w_k \bar{v}_k = \langle w, \bar{v}\rangle 
\end{align*}
the weighted average bid and ask volumes.
We define the imbalance as
\begin{equation*}
    \bar{\imath} : = \frac{B}{B+A} \in (0,1)
\end{equation*}
\begin{remark}
    In a blocked shaped setting, if $b$ denotes the amount of orders on every price level on the bid side, we get
    \begin{equation*}
        \bar{\imath} = \frac{\sum w_k b}{\sum  w_k (a+b)} = \frac{b}{a+b}
    \end{equation*}
    which yields $b = a\bar{\imath}/(1-\bar{\imath})$.
\end{remark}
The price deviation in the next period can be triggered by two events.
The posting and cancellation of incoming limit orders as well as the posting of market orders.
We distinguish between both, since the former does not have an impact on the execution of existing limit orders while the latter has.
We generically denote by
\begin{align*}
    d p& = \{\ldots, d p_{-1}, d p_0, d p_{1}, \ldots\}  & \mu & = \sum x dp_x\\
    dq & = \{\ldots, d q_{-1}, d q_0, d q_{1}, \ldots\}  & \nu & = \sum ydp_y\\
\end{align*}
the distribution and mean, respectively, of the two possible price movement in the next period.
For the sake of simplicity, we assume that the imbalance does not have an impact on incoming market orders and that the price deviation with respect to the market orders is neutral, that is $\nu =0$.
We furthermore assume that they are independent of each others.\footnote{The subsequent theoretical study adapt to eventual joint distribution of price movement due to limit and market orders also jointly dependent on the imbalance. The exposition of which is no longer explicit but can be solved numerically.}
To reflect the fact that the imbalance, as an indicator of the offer and demand on the market, has an impact on the market makers, we consider a parametrization $\imath \mapsto dp(\imath)$ of the price movement driven by limit orders as a function of the imbalance $\imath$.
Since the imbalance moves between $0$ and $1$, we assume that the distribution $dp$ moves as a convex combination of $\imath$ between the distribution $dp^-$ -- distribution when the imbalance is close to $0$, that is highly skewed to the left -- and the distribution $dp^+$ -- distribution when the imbalance is close to $1$, that is highly skewed to the right.
Mathematically:
\begin{equation*}
    dp(\imath) = \imath dp^+ + (1-\imath) dp^-
\end{equation*}
From the skewness assumptions and symmetry of the imbalance indicator, we assume that
\begin{equation*}
    dp^+_x \geq dp^-_x \quad \text{for every }x\geq 0 \quad \text{and}\quad dp^+_x = dp^-_{-x} \quad \text{for every }x
\end{equation*}
which implies that 
\begin{equation*}
    \mu(\imath) := \sum x dp_x(\imath) = (2\imath -1) \sum x dp^+_x = \mu^+(2\imath-1)
\end{equation*}
By assumption, $\mu^+$ is positive, showing that $\mu(i)$ moves between $-\mu^+$ and $\mu^+$ and is equal to $0$ for an imbalance of $1/2$ when offers equal demand.

\section{Spoofing strategy}
Suppose that at a given time we are given an ask price $p$ and a limit order book inventory $(\bar{v}^-, \bar{v})$.
A trader willing to buy an amount $H$ of shares faces the following three options.

\begin{itemize}[fullwidth]
    \item \textbf{Immediate market order:} for a total costs of 
        \begin{equation*}
            pH + \delta G(H)
        \end{equation*}
    \item \textbf{Delayed market order:} for an average total cost of 
        \begin{equation*}
            \sum_{x, y} \left[\left( p+\delta(x+y) \right) H +G(H)\right] dp_x\left( \bar{\imath} \right)dq_y = pH + \delta \left(G(H) + H \mu^+\left( 2\bar{\imath}-1 \right)\right)
        \end{equation*}
        Clearly if $\bar{\imath}<1/2$, then this second option is better than a direct buy.
    \item \textbf{Spoofing and delayed market order:}
        Book first an ask limit order $v$ at a depth $k$ in $\{0, 1,\ldots, N\}$ on top of the ask limit order book $\bar{v}_{k}$ to increase the liquidity on the ask side and signal a surge in supply to the market.
        In the next period the price deviates from $p$ to $p+\delta (x+y)$ and two situations may happen:
        \begin{itemize}
            \item $y\leq k$: no market order of sufficient magnitude hits the limit order book and therefore this limit order is not executed against an incoming market order.
                The previous limit order is canceled and the amount $H$ of shares is acquired for a cost of
                \begin{equation*}
                    \left( p+\delta (x+y) \right)H + \delta G(H)
                \end{equation*}
            \item $y>k$: the limit order is executed against an incoming market order at a price level $p+\delta(x+k)$.
                The new objective moves to $H+v$ resulting in a net cost of
                \begin{multline*}
                    \left( p+ \delta(x+y) \right)\left( H+v \right) + \delta G\left( H+v \right) - \left( p+\delta (x+k) \right)v \\
                    = \left( p+x \delta \right)H + \delta G(H+v) + \delta (y-k) v
                \end{multline*}
        \end{itemize}
        It follows that the spoofing net cost for a price deviation of $p+\delta (x+y)$ is given by
        \begin{equation}\label{eq:spoofing_costs}
            C_k(v,x,y) = (p+(x+y)\delta)H + \delta G\left(H + v 1_{\{y>k\}}\right) + \delta(y-k) v 1_{\{y>k\}}
        \end{equation}
        However, the posting of the selling limit order modifies the imbalance from $\bar{\imath}$ to 
        \begin{equation*}
            \imath_k(v) : = \frac{B}{A+B + w_k v}
        \end{equation*}
        In other words, the imbalance will move downwards, shifting the distribution $dp$ to more favorable outcomes.
        Since we assume that $\nu=0$, it follows that the average net costs are given by
        \begin{multline*}
            C_k(v):=\sum_{x,y} C_k(v,x, y) dp_x(\imath_k(v)) dq_{y}\\
            = \sum_{x} (p+\delta x)H dp_{x}(\imath_k(v)) + \delta(1-Q_k) G(H) + \delta Q_k G(H+v) + \delta v \sum_{y\geq k+1}(y-k)dq_y\\
            = pH + \underbrace{\delta (1-Q_k) G(H)}_{\text{Cost for optimal situation}}  + \underbrace{\delta H \mu^+ \left( 2\imath_k(v)-1 \right)}_{\text{Spoofing impact}} + \underbrace{\delta Q_k G(H+v)+ \delta v \nu_k}_{\text{Cost for being caught wrong way}} 
        \end{multline*}
        where
        \begin{equation*}
            Q_k = \sum_{y \geq k+1} dq_y \quad \text{and}\quad \nu_k =\sum_{y\geq k+1}(y-k)dq_y
        \end{equation*}
        Note that this cost functional is convex in $v$ since $G$ and $\imath_k$ are convex functions.
        Note also that in order to take advantage of this spoofing impact, it is necessary to drive the imbalance $\imath_k(v)$ below $1/2$.
\end{itemize}
\begin{remark}
    Note that we implicitly assume that the spoofing only happens at a given depth $k$.
    It is possible to spoof simultaneously at different depths resulting in a slightly more complex cost function that can be solved numerically.
    The conclusions do not change qualitatively and we use the more general multi-depth spoofing for the analysis of data in the subsequent sections.
\end{remark}
\subsection{Existence of Spoofing Manipulation}
The question is whether it is possible to push the imbalance as much as possible to $0$ in order to offset the costs of posting selling orders, they being executed and paying the liquidity costs of buying them back.
\begin{definition}
    We say that the limit order book $(\bar{v}^-, \bar{v})$ admits a (market taker) spoofing manipulation if there exists $v >0$ and $k \in \{0, 1, \ldots, N\}$ such that
    \begin{equation}\label{eq:spoofing}
        \begin{cases}
            \displaystyle C_k(v) < pH +\delta \left(G(H) +  \mu^+ H\left(2\bar{\imath} -1\right)\right) & \text{if } \ \bar{\imath}\leq 1/2\\
            \\
            \displaystyle C_k(v) < pH +\delta G(H)  & \text{if } \ \bar{\imath}> 1/2
        \end{cases}
    \end{equation}
\end{definition}
According to the average costs of spoofing, these two inequalities turn into
\begin{equation}\label{eq:spoofing02}
    Q_k\left[G(H+v) - G(H) \right]< 2\mu^+ H\left(\bar{\imath}\wedge \frac{1}{2} -\imath_k(v)\right) - v \nu_k
\end{equation}

The following results concerns the existence of spoofing manipulation in a blocked shaped setting where the volume on the ask side of the limit order book is $a$ everywhere.
For scaling reasons let $H=\rho a$, where $\rho$ represents the ratio of the shares to purchase to the depth of the limit order book.
Our first result concerns the existence of spoofing manipulation.
\begin{proposition}\label{prop:existence_spoofing}
    In a block shaped setting, where the volume on the ask side of the limit order book is $a$ everywhere.
    The following assertions hold
    \begin{itemize}[fullwidth]
        \item If $\bar{\imath}\leq 1/2$, the limit order book admits no spoofing manipulation if and only if \eqref{eq:NS} holds;
        \item If $\bar{\imath}> 1/2$, the limit order book admits no spoofing manipulation if \eqref{eq:NS} holds.
    \end{itemize}
    Where $H=\rho a$ and
    \begin{equation}\label{eq:NS}
        2 \rho \mu^+\left( 1-\bar{\imath} \right)\bar{\imath} w_k \leq Q_k \rho +\nu_k \quad \text{for all }k
    \end{equation}
\end{proposition}
For the proof, see Appendix \ref{appendix:proofs}.
From this proposition, we deduce that price manipulation is more likely to occur
\begin{itemize}
    \item if $Q_k$ is small -- and as a byproduct $\nu_k$.
        If the probability to get a spoofing order executed is small, there is relatively no downsize at spoofing.
    \item if $\bar{\imath}$ is close to $1/2$.
        If the imbalance is close to $1/2$, then $\bar{\imath}(1-\bar{\imath})$ is maximum.
        The impact of moving the price in ones favor is maximal there.
    \item if $\mu^+$ is large:
        $\mu^+$ represents the mean deviation sensitivity as a function of the imbalance.
        The more sensitive the price movement is with respect to the imbalance, the more likely spoofing strategies may occur.
    \item If $w_k$ is large:
        $w_k$ represents the relative impact at tick level $k$ of a spoofing volume to the imbalance.
        If one of $w_k$ is large with respect to the corresponding $Q_k$, then spoofing is more likely to occur there.
    \item if $\rho$ is relatively large.
        If the amount of order to buy relative to the overall offer is very large, spoofing is more likely to happen.
\end{itemize}
Figure \ref{fig:condition_function} represents the existence of spoofing condition \ref{eq:NS} in terms of the initial imbalance $\bar{\imath}$ with varying market parameters.

\begin{figure}[H]
    \centering
    \includegraphics[width=\textwidth,keepaspectratio]{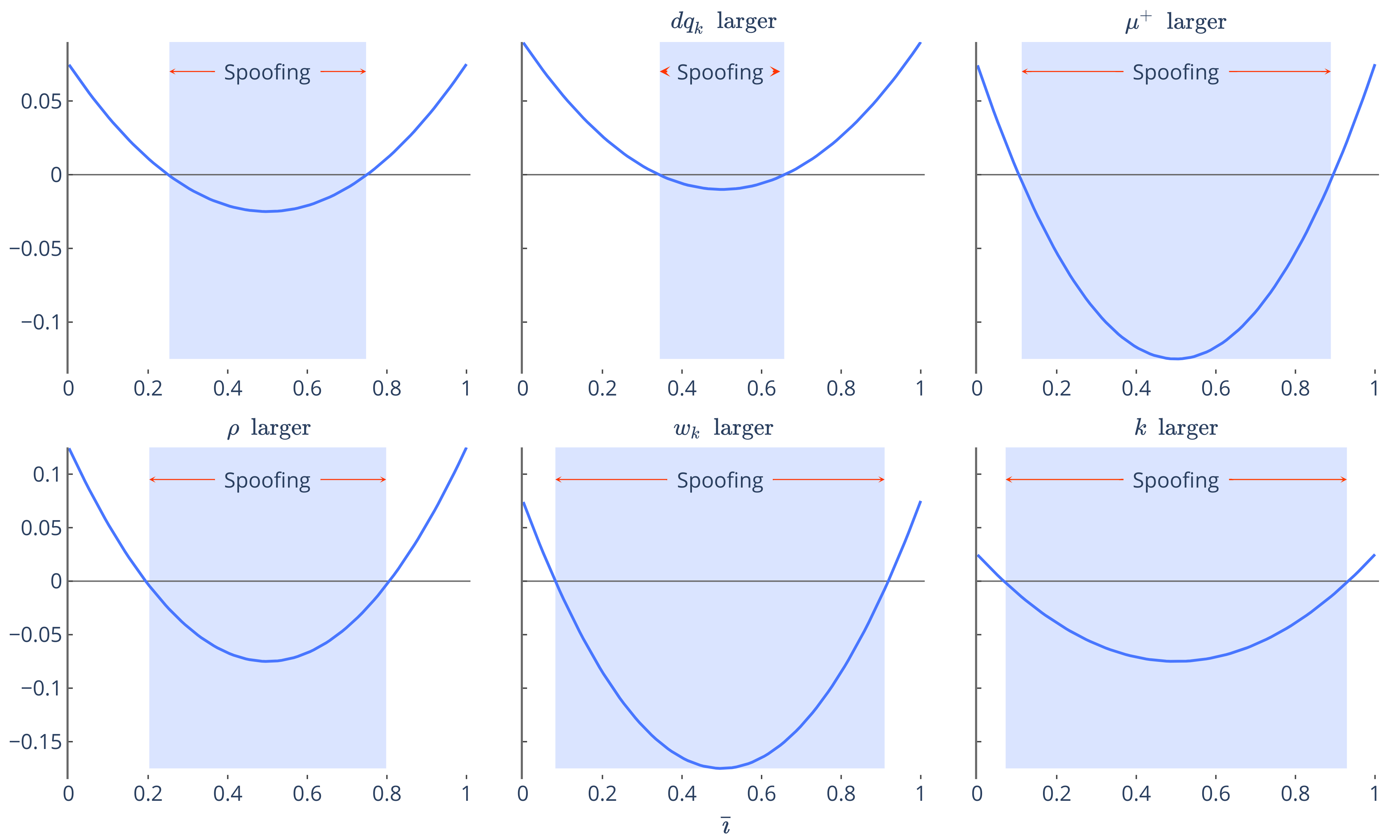}
    \caption{Spoofing condition (\ref{eq:NS}) as a function of $\bar{\imath}$  and in $(a)$
    $\mu^ + =1, \rho = 1, k = 3, w_k =0.2,  dq_y = 0.025 $ for all $y \geq k $.
    One parameter is increased each time with respect to $(a)$ where 
    $(b)$: $dq_y = 0.03 $ for all $y \geq k $; 
    $(c)$: $\mu^+ =2$;
    $(d)$: $\rho =2$;
    $(e)$: $w_k = 0.5$;
    $(f)$: $k=4$.}
    \label{fig:condition_function}
\end{figure}  

\subsection{Optimal Spoofing and Resulting Imbalance Impact}
Let us now address the problem of finding the optimal spoofing strategy.
In particular as a function of the depth at which the spoofing order is placed.
\begin{proposition}\label{prop:optimal_spoofing}
    The optimal spoofing volume $v_{spoof}$ at a given level $k$ and resulting imbalance $\imath_{spoof}$ -- adopting the notations $w := w_k$, $Q := Q_k$ and $\nu:=\nu_k$ -- are given by:
    \begin{equation*}
        v_{spoof} = \frac{a}{Q}\left[2\rho w \mu^+ \frac{1-\bar{\imath}}{\bar{\imath}} \imath_{spoof}^2 - \left( Q\rho +\nu \right)\right]^+
    \end{equation*}
    where $\imath_{spoof}$ is the unique cubic root solution in $(0, \bar{\imath}]$ of
    \begin{equation*}
        \frac{\bar{\imath}}{\imath} = 1+\frac{(1-\bar{\imath})w}{Q}\left[2\rho w \mu^+ \frac{1-\bar{\imath}}{\bar{\imath}} \imath^2 - \left( Q\rho +\nu \right)\right]^+
    \end{equation*}
\end{proposition}
For the proof, see Appendix \ref{appendix:proofs}.
Though the solution is implicit, we can inspect the spoofing behavior as a function of the distance to the top of the limit order book.
Note first that $\imath_{spoof} = \bar{\imath}$ if and only if $2\rho w\mu^+ (1-\bar{\imath}) \bar{\imath}  \leq  Q\rho +\nu $ which results into $v_{spoof} = 0$.
This coincide with the no spoofing condition of the previous proposition.
We are interested in the relative spoofing size as a function of these parameters.
From the definition of the imbalance, $\imath(v)$ increases if and only if $v$ decreases, so we get more spoofing volume as $\imath_{spoof}$ gets smaller.
Now from the implicit function it holds that
\begin{equation*}
    \frac{1}{\imath} = \frac{1}{\bar{ \imath}} + \frac{w}{Q} \frac{1-\bar{\imath}}{\bar{\imath}} \left[ 2\rho w\mu^+ \frac{1-\bar{\imath}}{\bar{\imath}} \imath^2  - \left( Q\rho +\nu \right) \right]^+=:f\left(w, \mu^+, \nu, Q, \bar{\imath}, \rho,\imath \right)
\end{equation*}
where the function $f$ is an increasing function of $\imath$ greater than $1 / \bar{\imath}$.
\begin{itemize}
    \item Since $f$ is increasing as a function of $w$ and $\mu^+$, it follows that $\imath_{spoof}$ is decreasing as a function of $w$ and $\mu^+$.
        Hence, spoofing behavior increases as a function of the impact $w$ at level $k$ on the imbalance as well as a function of the overall sensitivity $\mu^+$ of the price movement with respect to the imbalance.
    \item Since $f$ is decreasing as a function of $Q$ and $\nu$, it follows that $\imath_{spoof}$ is increasing as a function of $Q$ and $\nu$.
        From an empirical viewpoint, $Q=Q_k$ as well as $\nu=\nu_k$ decreases as a function of the depth $k$.
        It follows that spoofing behavior is more likely to happen and increase deeper in the limit order book.
        However, this conclusion is short of the fact that local sensitivity of the imbalance $w=w_k$ also depends on the depth with an inverse impact.
        According to empirical analysis, it turns out that $w$ does not exert this decreasing behavior as a function of $k$ at least within a reasonable depth in the limit order book.
        It seems that spoofing behavior is more likely to happen at a reasonable distance from the top of the limit order book. 
\end{itemize}
The behavior of the resulting imbalance $\imath_{spoof}$ as a function of the initial imbalance $\bar{\imath}$ is more difficult to stress out.
We know that $\imath_{spoof} \leq \bar{\imath}$ and for the same reasons as before it is increasing as a function of $\bar{\imath}$.
The same holds for the dependence on the relative number of shares to purchase $\rho$.
Figure \ref{fig:optimal_spoofing} represents the curves of the spoofed imbalance $\imath_{spoof}(\bar{\imath})$ as a function of the initial imbalance $\bar{\imath}$ for different depths with varying market parameters.

\begin{figure}[H]
    \centering
    \includegraphics[width=\textwidth,keepaspectratio]{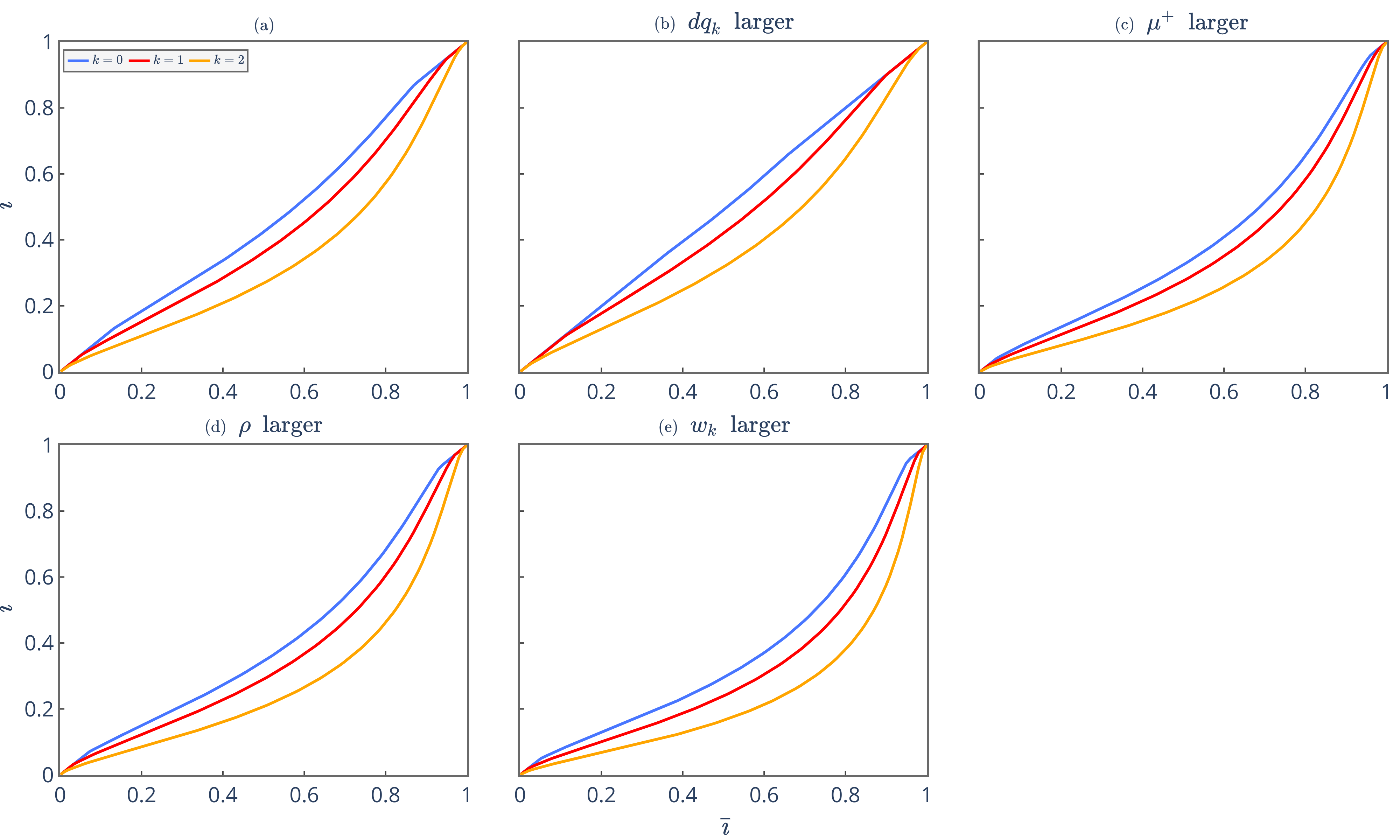}
    \caption{$\imath$ as a function of $\bar{\imath}$ and in $(a)$
    $\mu^ + =1, \rho = 1, w_k =0.2,  dq_y = 0.003 $ for all $y \geq k $.
    One parameter is increased each time with respect to $(a)$ where 
    $(b)$: $dq_y = 0.006 $ for all $y \geq k $; 
    $(c)$: $\mu^+ = 3$;
    $(d)$: $\rho =3$;
    $(e)$: $w_k = 0.5$;
    Blue line :$k=0$; red line: $k=2$; orange line: $k=4$.}
    \label{fig:optimal_spoofing}
\end{figure}

\begin{remark}
    Throughout, we mainly focus on the spoofing behavior from a market taker's viewpoint.
    As for a market maker, spoofing behavior might be rewarding as well.
    It turns out that the rewards from spoofing are intertwined with the ones from pure market making.
    The resulting impact on the imbalance is however quite similar, up to the fact that the bid ask spread is an additional factor in the spoofing opportunity, since the market maker will have to cross the spread.
    In Appendix \ref{appendix:round_trip}, we derive the spoofing strategy in the round trip context, discuss the spoofing impact on the imbalance and numerical analysis in the same context as the present situation.
\end{remark}

\section{Calibration}
According to this model, we calibrate the imbalance as well as $dp$ and $dq$ on real data provided by TMX.
These datasets consists of level 2 data from June to September 2017.
Among the 1500 available equities we selected 10, varying in company background, market capitalization as well as trading frequency.
The level 2 datasets include time, order price, volume, type,\footnote{Buy/sell; booked/cancelled/traded} order ID and counterpart order ID in case of a trade, see Figure \ref{fig:LOB_RAW}
\begin{figure}[H]
    \centering
    \includegraphics[width=\textwidth,keepaspectratio]{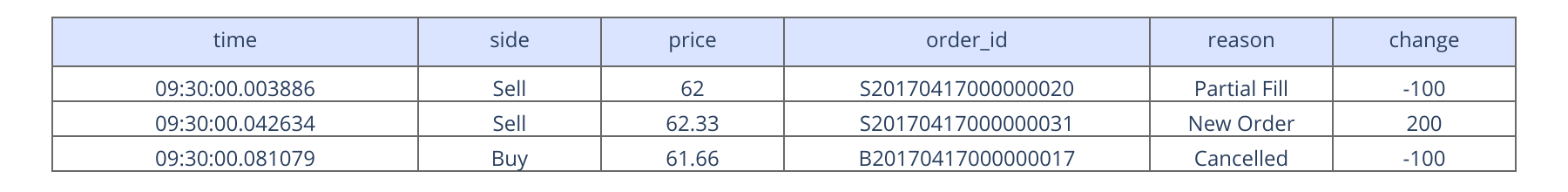}
    \caption{Original Level 2 dataset of stock AEM provided by TMX.}
    \label{fig:LOB_RAW}
\end{figure}
Since these are provided in diff form, therefore a cumulative aggregation allows to construct the full limit order book at any time as in Figure \ref{data_pic}
\begin{figure}[H]
    \centering
    \includegraphics[width=\textwidth,keepaspectratio]{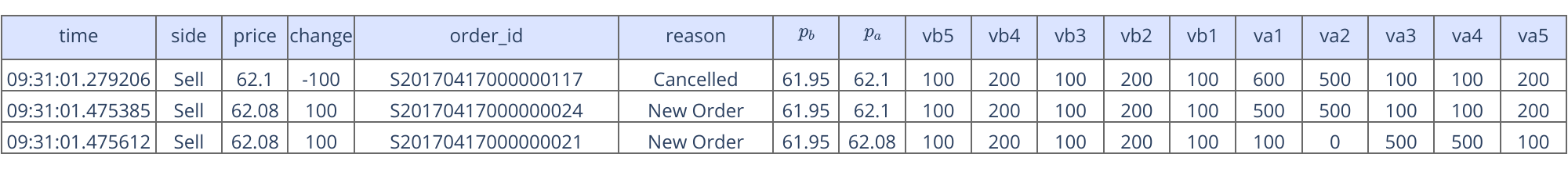}
    \caption{Generation of the full limit order book out of the Level 2 data.}
    \label{data_pic}
\end{figure}
This operation is computationally very intensive and therefore has been realized on a distributed data cluster of TMX with spark.

With the full limit order books at hand we divide the calibration into three steps:
\begin{itemize}
    \item Find a normalized sample frequency and choose a maximal depth for the support of the distributions $dp^\pm$ and $dq$;
    \item Calibrate the imbalance generically, that is, as a function of the weights $w=(w_1, w_2, \ldots)$;
    \item Estimate $dq$, $dp^\pm$ and weights $w=(w_1, w_2, \ldots)$.
\end{itemize}
The sample frequency should be large enough such that there exists enough variance in price change, see Figure \ref{AEM_price_change}
\begin{figure}[H]
    \centering
    \includegraphics[width=\textwidth,keepaspectratio]{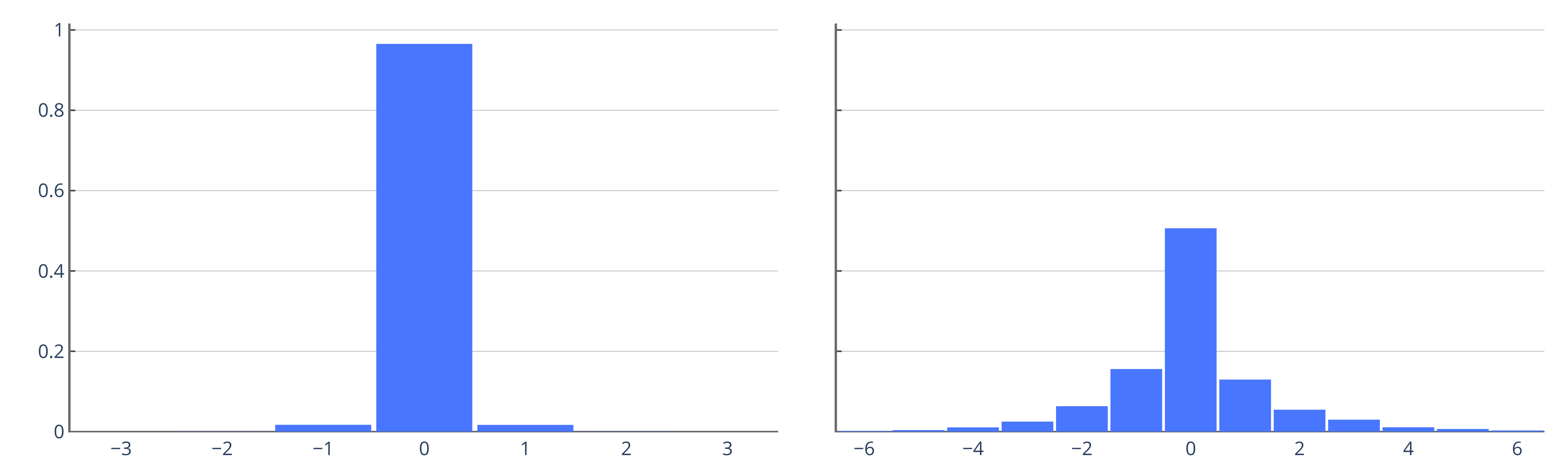}
    \caption{Left panel: Histogram of original AEM price change. Right panel: Histogram of AEM price change after sampling.}
    \label{AEM_price_change}
\end{figure}
To compare across markets with different trading activity---and eventually time during the day---we fix a target variance of $\sigma^2$ for the price movement and select the optimal frequency $f$ for each stock as to minimize the square distance between $\sigma_f$ and $\sigma$.
For a target variance $\sigma^2 =2$, Table \ref{table:stats_stocks} is the sample frequency for different stocks with a summary statistics of the average volume and arrival rate for Market/Limit Orders on the bid and ask side.
As for the maximal depth for the support of the distribution, we take the 99\% quantile of empirical sampled price change distribution.
The depths are around 4 since we use the same $\sigma^2 = 2$ to determine the sampling frequency. The sampling frequency is related to how fast limit orders arrive, not market orders. For most stocks, $f$ is small when the arrival rate of limit orders is high.
\begin{table}[H]
    \begin{center}
        \resizebox{\columnwidth}{!}{
            \begin{tabular}{@{}crcrrcrrcrrcrr@{}}
                \toprule
                \multirow{3}{*}{Stock} & \multirow{3}{*}{$f$ } & \multirow{3}{*}{Depth} 
                                       & \multicolumn{5}{c}{Market Orders} & & \multicolumn{5}{c}{Limit Orders}  \\
                                       \cmidrule{4-8} \cmidrule{10-14}
                                       & & & \multicolumn{2}{c}{Buy} & & \multicolumn{2}{c}{Sell} & & \multicolumn{2}{c}{Buy} & & \multicolumn{2}{c}{Sell} \\
                                       \cmidrule{4-5} \cmidrule{7-8} \cmidrule{10-11} \cmidrule{13-14}
                                       & & & \multicolumn{1}{c}{Vol} & \multicolumn{1}{c}{Rate} & & \multicolumn{1}{c}{Vol} &  \multicolumn{1}{c}{Rate} & & \multicolumn{1}{c}{Vol} & \multicolumn{1}{c}{Rate} & & \multicolumn{1}{c}{Vol} &  \multicolumn{1}{c}{Rate}  \\
                                       \midrule
                AEM & 4  & 4 & 146 & 0.072 &  & 144 & 0.065 &  & 144  & 5.145 &  & 142  & 5.232    \\
                BB  & 38 & 4 & 567 & 0.107 &  & 613 & 0.085 &  & 3257 & 3.034 &  & 3245 & 3.008    \\
                BMO & 11 & 4 & 171 & 0.107 &  & 168 & 0.134 &  & 149  & 2.482 &  & 151  & 2.553    \\
                CNR & 6  & 4 & 141 & 0.104 &  & 134 & 0.1   &  & 138  & 2.397 &  & 141  & 2.399  \\
                CPG & 53 & 5 & 356 & 0.13  &  & 349 & 0.115 &  & 866  & 1.762 &  & 879  & 1.741     \\
                FNV & 3  & 5 & 123 & 0.059 &  & 121 & 0.056 &  & 126  & 2.881 &  & 140  & 2.794    \\
                FR  & 60 & 3 & 243 & 0.072 &  & 272 & 0.063 &  & 795  & 2.115 &  & 842  & 2.049   \\
                PPL & 26 & 4 & 142 & 0.101 &  & 152 & 0.11  &  & 165  & 2.31  &  & 183  & 2.469   \\
                TD  & 20 & 4 & 223 & 0.205 &  & 213 & 0.218 &  & 277  & 3.87  &  & 278  & 3.866    \\
                VET & 6  & 5 & 127 & 0.07  &  & 152 & 0.069 &  & 130  & 2.474 &  & 137  & 2.439   \\
                \bottomrule
            \end{tabular}
        }
        \caption{Stock data from June 5, 2017 to June 9, 2017. The Vol columns corresponds to the average volume of a single order during inspected time interval and the Rate columns corresponds to the number of orders per second.}
        \label{table:stats_stocks}
    \end{center}
\end{table} 

With the sampling frequency $f$ and depth $N$, we define the average imbalance at time $t$ as 
\begin{equation*}
    \hat{\imath}(w, t) 
    = \frac{
        \sum\limits_{k\leq N} \sum\limits_{t - f  \leq s<t}  \bar{v}^-_k(s) w_k
        }
        {\sum\limits_{k\leq N} \sum\limits_{t - f  \leq s<t}  \left(\bar{v}^-_k(s) + \bar{v}_k(s) \right) w_k
        }
\end{equation*}
which sums up order book volumes within a certain time interval weighted by time difference $\Delta s_i = s_{i+1} - s_{i}$.
The weighting parameter $w$ impacts the average imbalance distribution, see Figure \ref{BMO_skew_histo_diff}.
\begin{figure}[H]
    \centering 
    \includegraphics[width=0.8\linewidth]{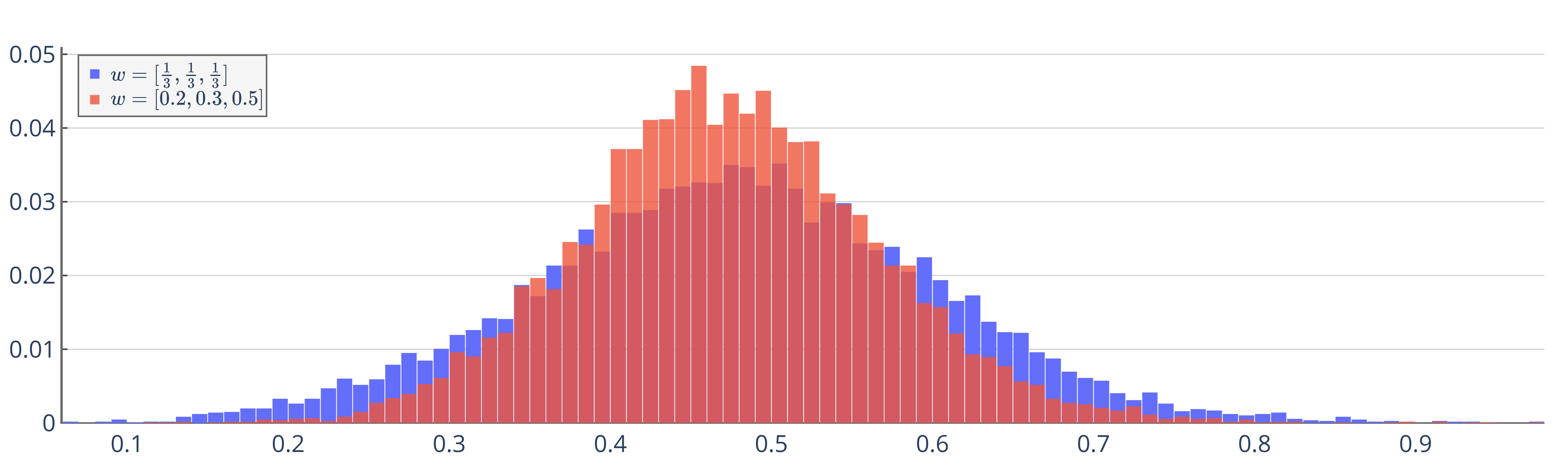}
    \caption{Distribution of the average imbalance for different weights for the stock BMO.} 
    \label{BMO_skew_histo_diff}
\end{figure} 
Nevertheless for each weight vector, the resulting distribution is close to a skewed normal distribution.
For a given weight $w$, using maximum likelihood, we fit the empirical distribution to the corresponding skew normal distribution $\mathcal{SN}(\alpha(w), \xi(w), \omega(w))$, the fit of which is particularly good, see figure \ref{BMO_skew_histo} for an example.
\begin{figure}[H]
    \centering 
    \includegraphics[width=0.8\linewidth]{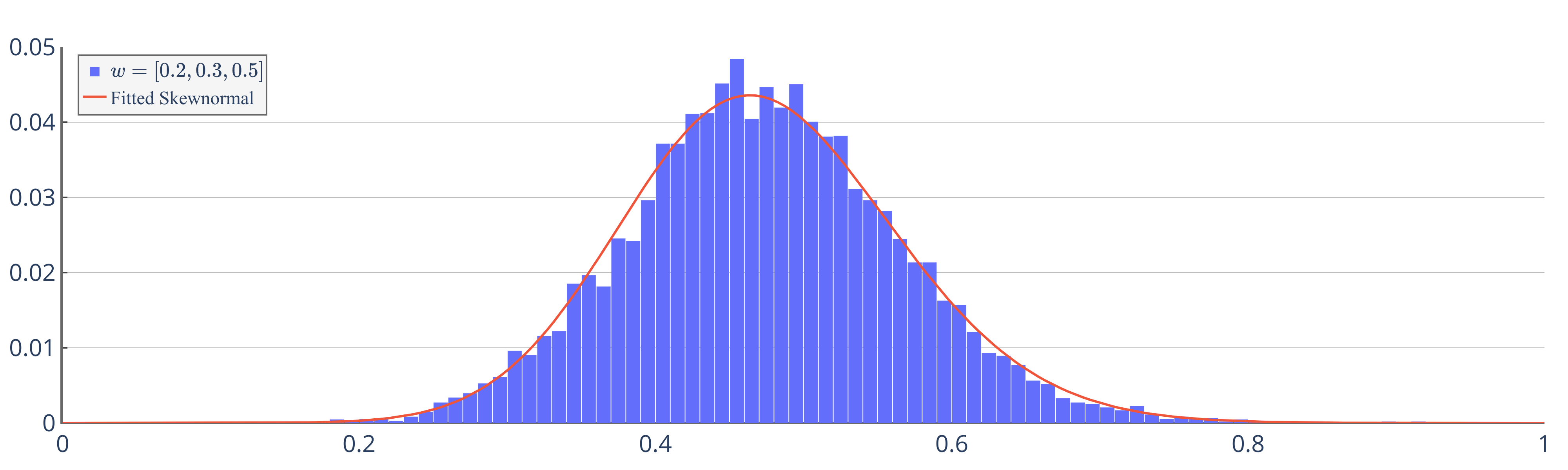}
    \caption{Histogram of BMO average imbalance and fitted skewnormal distribution: $w = [0.2,0.2,0.2,0.2,0.2]$.} 
    \label{BMO_skew_histo}
\end{figure} 

The third step is to determine $dq$, $dp^\pm$ with the optimal weights $w$.
As for $dq$, it is the probability that the price moves by $k$ ticks triggered by market orders.
Thus for each market order, compute
\begin{equation*}
    F(H) = \inf \{ x\in \mathbb{N}_0 : \sum_{k=0}^x v_k \geq H\}
\end{equation*}
where $H$ is the volume of the market order.
This represents exactly how many positive tick price deviation an order of size $H$ will produce.
We derive $dq$ from the empirical distribution.

As for $dp^\pm$ and $w$, a maximum likelihood estimation is implemented to solve
\begin{equation}
    dp^\ast , w^\ast = \argmin_{dp^+,w} \left[-\frac{1}{M} \sum^M_{m=1} \log p\left(x_m,  \hat{\imath}_m \right)\right]
\end{equation}
where $x_m$ is the empirical price change, $\hat{\imath}_m$ the average imbalance for a given weight $w$, 
\begin{equation*}
    p\left(x_m,  \hat{\imath}_m \right) = dp_{x_m}\left(\hat{\imath}_m \right) p(\hat{\imath}_m)
\end{equation*}
where $ dp_{x_m}\left(\hat{\imath}_m\right) = \hat{\imath}_m dp^+_{x_m} + \left( 1-\hat{\imath}_m \right) dp^-_{x_m}$ represents the conditional probability of price change equal to $x_m$ given $\hat{\imath}_m$, and $p(\hat{\imath}_m)$ is the density of the fitted skewnormal distribution evaluated at $\hat{\imath}_m$.

Figure \ref{stock_weights}, illustrating the value of the optimal weights $w$ for selected stocks, shows different patterns.
Overall, it turns out that the relative impact of the imbalance to the price distribution is more important away from the top of the limit order book.

We also performed this calibration procedure on stock BMO weekly from June 5th to June 30th, as well as for the first hour of trading monthly from June to September. Figure \ref{BMO_weights} provide the optimal weights in each case for BMO.
\begin{figure}[H]
    \centering
    \includegraphics[width=\textwidth,keepaspectratio]{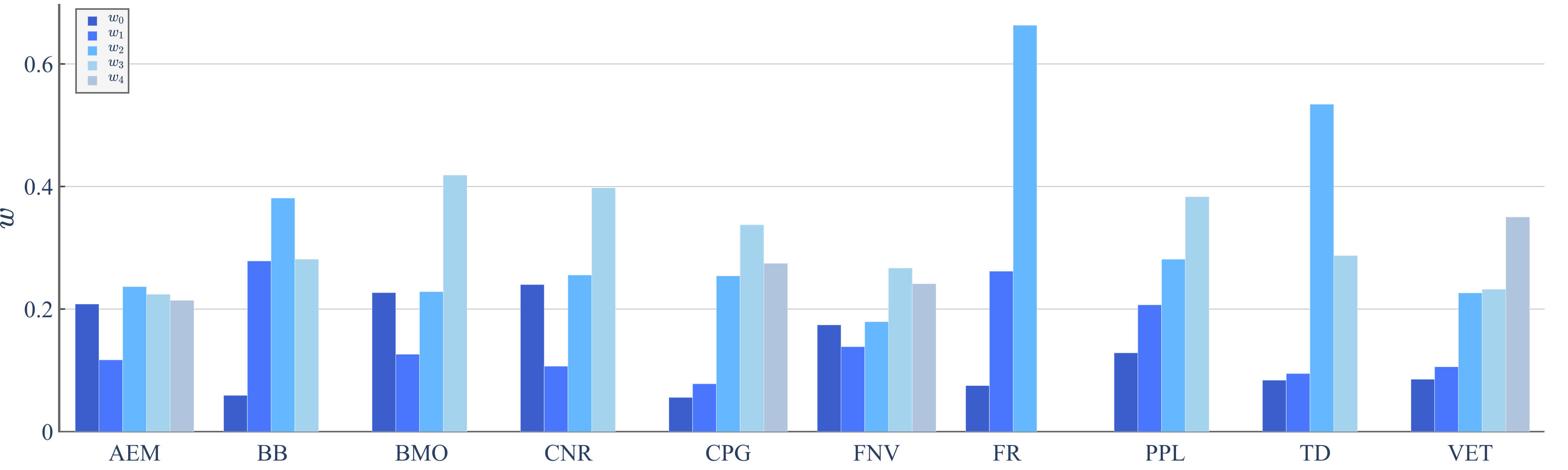}
    \caption{$w$ for stock AEM, BB, BMO, CNR, CPG, FNV, FR, PPL, TD, VET from June 5th, 2017 to June 9th, 2017.}  
    \label{stock_weights}
\end{figure}

\begin{figure}[H]
    \centering
    \includegraphics[width=\textwidth,keepaspectratio]{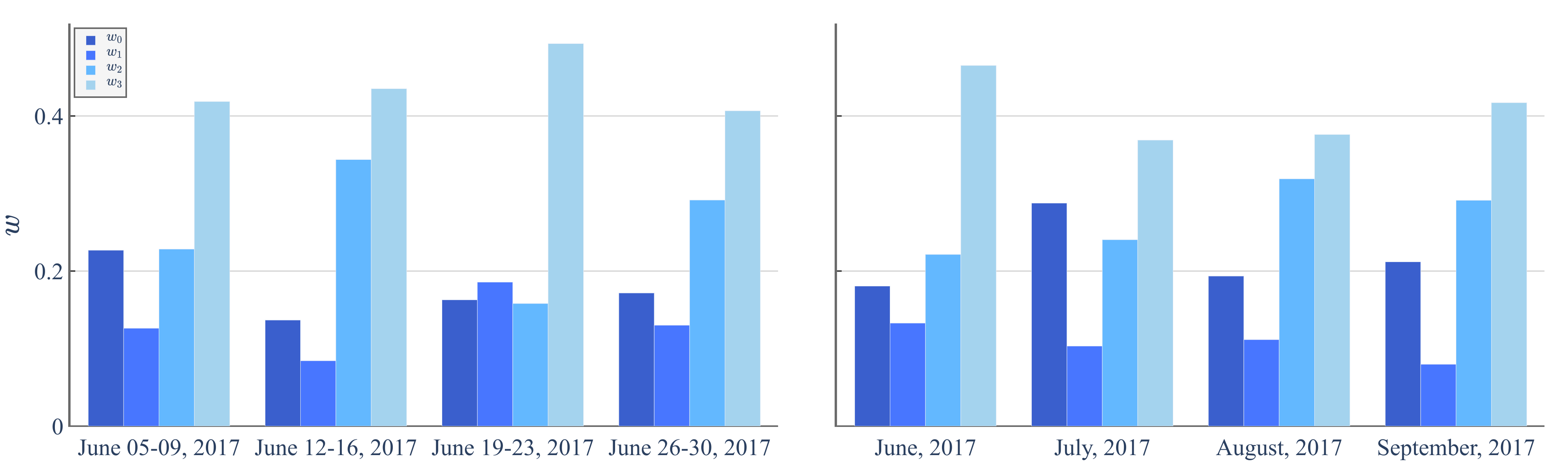}
    \caption{Left panel: $w$ for stock BMO each week in June 2017. Right panel: $w$ for stock BMO each month in June 2017, only using first hour trading data.}  
    \label{BMO_weights}
\end{figure}
Notice that for the first hour of trading the optimal weights are more consistent across time, but all show that the weight impact on the price movement happens deeper in the limit order book.

As for the corresponding $dp^+$ and $dq$, they are represented in Figure \ref{dp_dq} for stock BMO from June 5th, 2017 to June 9th, 2017.
As expected, $dp^+$, representing the price movement as the imbalance is large, is skewed to the right.
\begin{figure}[H]
    \centering
    \includegraphics[width=\textwidth,keepaspectratio]{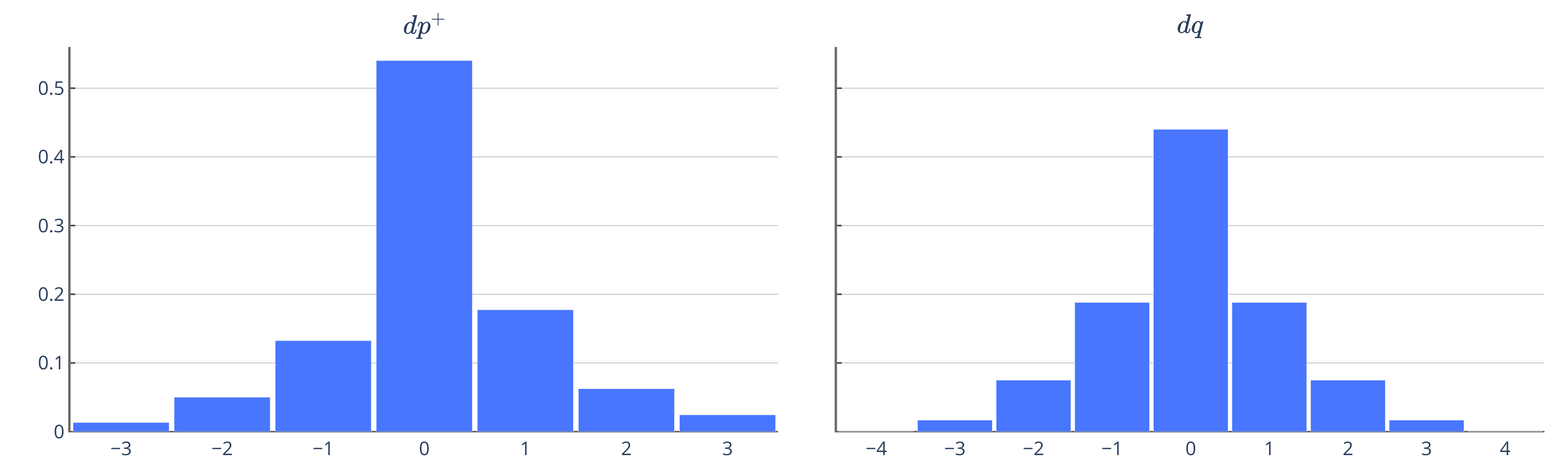}
    \caption{$dp^+$ and $dq$ for BMO from June 5th to June 9th.} 
    \label{dp_dq}
\end{figure}
Table \ref{table:moments} provides the moments of $dp^+$ for each stock -- in tick values.
Skewness shows how much $dp^+$ is skewed to right.
When it is large, $\mu^+$ is also large and spoofing has a larger impact according to theoretical part.
Except for CPG which is relatively small, all of the other stocks under study excerpt this pattern of right-skewness.
\begin{table}[H]
    \begin{center}
            \begin{tabular}{@{}lrrrr@{}}
                \toprule
                \multicolumn{1}{c}{Stock}                   & \multicolumn{1}{c}{$\mu^+$} & \multicolumn{1}{c}{Variance} & \multicolumn{1}{c}{Skewness} & \multicolumn{1}{c}{Kurtosis}  \\
                \midrule
                AEM & 0.411 & 0.881 & 1.175  & 4.46
                \\
                BB  & 0.467 & 0.766 & 0.957  & 3.786
                \\
                BMO & 0.103 & 1.087 & 0.095  & 4.263
                \\
                CNR & 0.398 & 0.636 & 1.447  & 4.928
                \\
                CPG & 0.100 & 1.171 & -0.068 & 5.275
                \\
                FNV & 0.404 & 0.850 & 2.038  & 7.169
                \\
                FR  & 0.209 & 1.033 & 0.532  & 2.203
                \\
                PPL & 0.076 & 1.177 & 0.645  & 3.83
                \\
                TD  & 0.118 & 1.389 & 0.633  & 4.471
                \\
                VET & 0.315 & 1.253 & 1.899  & 7.393
                \\
                \bottomrule
            \end{tabular}
        \caption{Moments of $dp^+$.}
        \label{table:moments}
    \end{center}
\end{table}

\section{Approaches to Spoofing Detection}
For reasons mentioned in the introduction, it is difficult from a regulatory viewpoint to figure out whether or not spoofing happened a-posteriori.
According to the theoretical part, the act of spoofing will influence the resulting imbalance.
However, to monitor the imbalance is akin to contemplate pure noise as shown in Figure \ref{instant_imbalance}.
\begin{figure}[H]
    \centering
    \includegraphics[width=\textwidth,keepaspectratio]{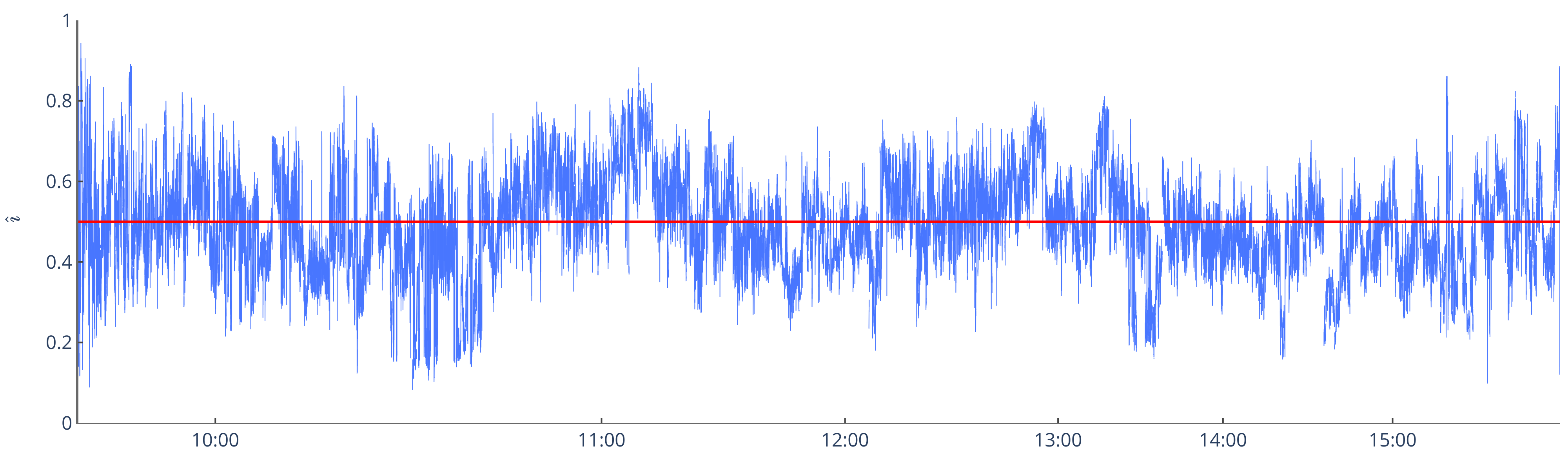}
    \caption{Imbalance of stock BMO from 09:30 to 16:00 on June, 7, 2017.} 
    \label{instant_imbalance}
\end{figure}

In the following, we propose some possible ways to perform such a monitoring based on the theoretical results.
The strategy comes from the following observation:
For a spoofing strategy to be successfully fulfilled, a market order has to be executed.\footnote{In this paper, we do not consider spoofing strategies involving only limit orders.}
Hence when observing an executed market order two situations may happen:
\begin{enumerate}[label=\textit{\arabic*-}]
    \item The market order is a legitimate one.
        In that case, the imbalance before this market order $\imath_-$ and after $\imath_+$ should follow statistically the classical long run behavior.
        In other words, in a legitimate situation, we should observe statistically the pair 
        \begin{equation*}
            (\imath_-, \imath_+)
        \end{equation*}
        for each market order.
    \item The market order is the result of a spoofing behavior.
        The implicit equilibrium without spoofing would be the pair $(\imath_-, \imath_+)$.
        After the market order is executed, the market imbalance should be back to its equilibrium $\imath_+$.
        However before the market order, the spoofer observes the implicit imbalance $\imath_-$ and decides to spoof according to this information, sending to the market $\imath_{spoof}(\imath_-)$ instead of $\imath_-$.
        The resulting observation for those spoofed market orders is therefore the pair
        \begin{equation*}
            (\imath_{spoof}(\imath_-), \imath_+)
        \end{equation*}
\end{enumerate}
Furthermore, spoofing strategies are supposed to happen sporadically but intensively within a short time horizon.
Before presenting some strategies, let us fix some notations:
\begin{itemize}
    \item $\Pi = \{t_1< t_2 <\ldots < t_M\}$ represents the time stamps of each (buy) market orders in a long sample (several weeks).
    \item $\hat{\imath}_-(t)$ and $\hat{\imath}_+(t)$ represents the imbalance before and after the market order happening at time $t$ in $\Pi$.
    \item $(\imath_-, \imath_+)$ represents the overall joint distribution of the imbalance right before and after each market orders fitted to the overall data.
        We assume that these represents the stable behavior of the market without spoofing, and therefore representative of legitimate market orders.
    \item $(\hat{\imath}_-^N(t), \hat{\imath}^N_+(t))$ represents the (short span) empirical distribution at time $t$ in $\Pi$ generated by the last $N$ market orders observed imbalances $(\hat{\imath}_-(s), \hat{\imath}_+(s))$, where $N\ll M$ is a short horizon sample size (in our case about 100).
    \item The previous theoretical part, even if not explicit in terms of solution allows us to compute numerically $\imath_{spoof}(\imath_-)$ for a given implicit imbalance $\imath_-$.
\end{itemize}

\subsection{Monitoring $\hat{\imath}^N_-$}
A first idea is to monitor the behavior of the short term imbalance $\hat{\imath}^N_-(t)$ as times passes to test whether it is statistically different from the equilibrium $\imath_-$.
This is however not adequate for the following reasons.
First, this is not related to spoofing behavior and might reflects some other market patterns.
Second, and more importantly, the sequence of $\hat{\imath}_{-}(t)$ for each market order is highly dependent.
Indeed, there might exists market conditions -- bullish/bearish, etc -- such that a short horizon sample $\hat{\imath}^N_-$ differs strongly from the long term behavior.
Figure \ref{Autocorrelation} provides empirical evidence about the sequential dependence of the imbalance $\hat{\imath}_-$ as well as $\hat{\imath}_+$ over time.
\begin{figure}[H]
    \centering
    \includegraphics[width=\textwidth,keepaspectratio]{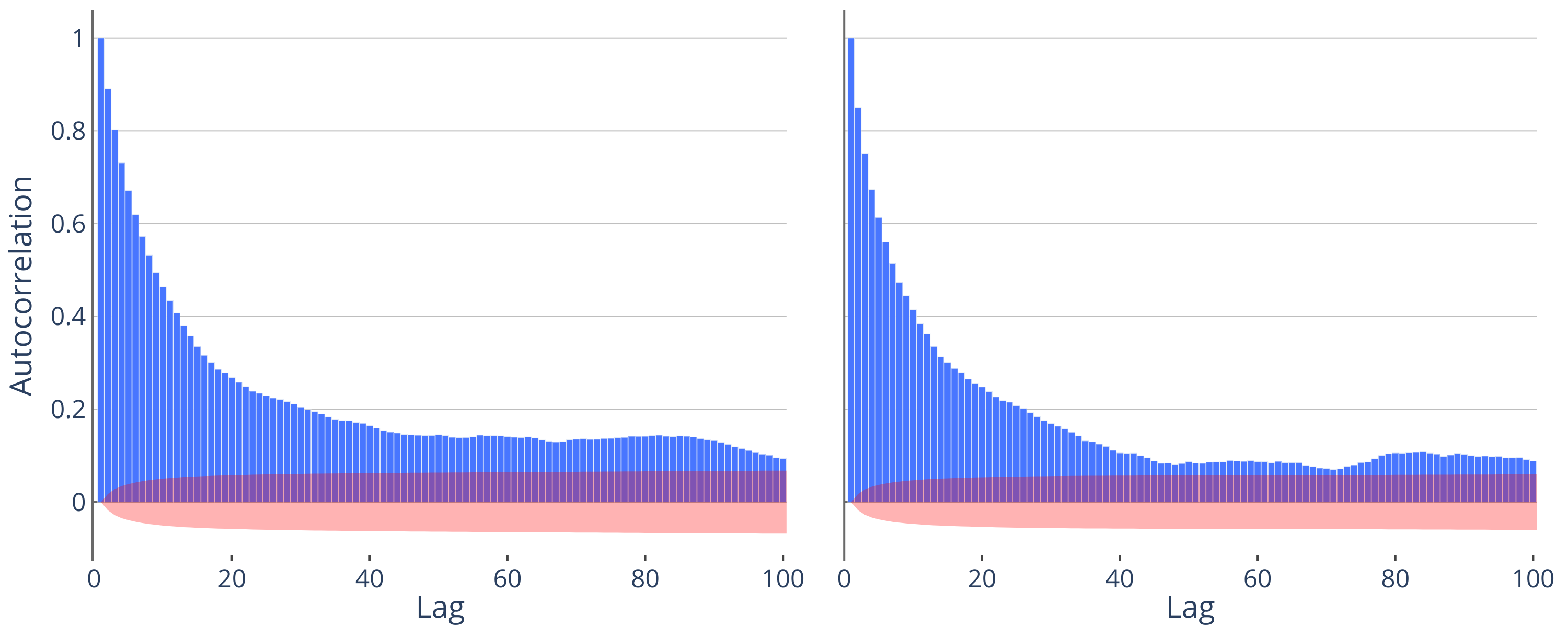}
    \caption{Left panel: $\hat{\imath}_-$ autocorrelation of stock BMO on June 7, 2017. Right  panel: $\hat{\imath}_+$ autocorrelation of stock BMO on June 7, 2017. Red area is the 95\% confidence interval of the autocorrelation.} 
    \label{Autocorrelation}
\end{figure}

\subsection{Monitoring $(\hat{\imath}^N_-, \hat{\imath}^N_+)$}
The statistical link towards discrimination of $\imath_{spoof}(\imath_-)$ from $\imath_-$ is the additional observation of the imbalance after the spoofing happen.
This provides statistical a-posteriori information about the implicit market equilibrium before spoofing which in case of spoofing can not be directly observed.
Figure \ref{joint_dis} shows on the left panel the joint distribution $(\imath_-, \imath_+)$ while the right panel represents, based on the model of the theoretical part and calibration, the joint distribution $(\imath_{spoof}(\imath_-), \imath_+)$ in the case of spoofing.
The spoofed joint distribution is skewed to the left in comparison to the non-spoofed one, in accordance to the theoretical analysis that spoofing decreases the imbalance -- in the buy order case -- before a market order.
\begin{figure}[H]
    \centering
    \begin{minipage}[c]{0.5\textwidth}
        \centering
        \includegraphics[width=\textwidth,keepaspectratio]{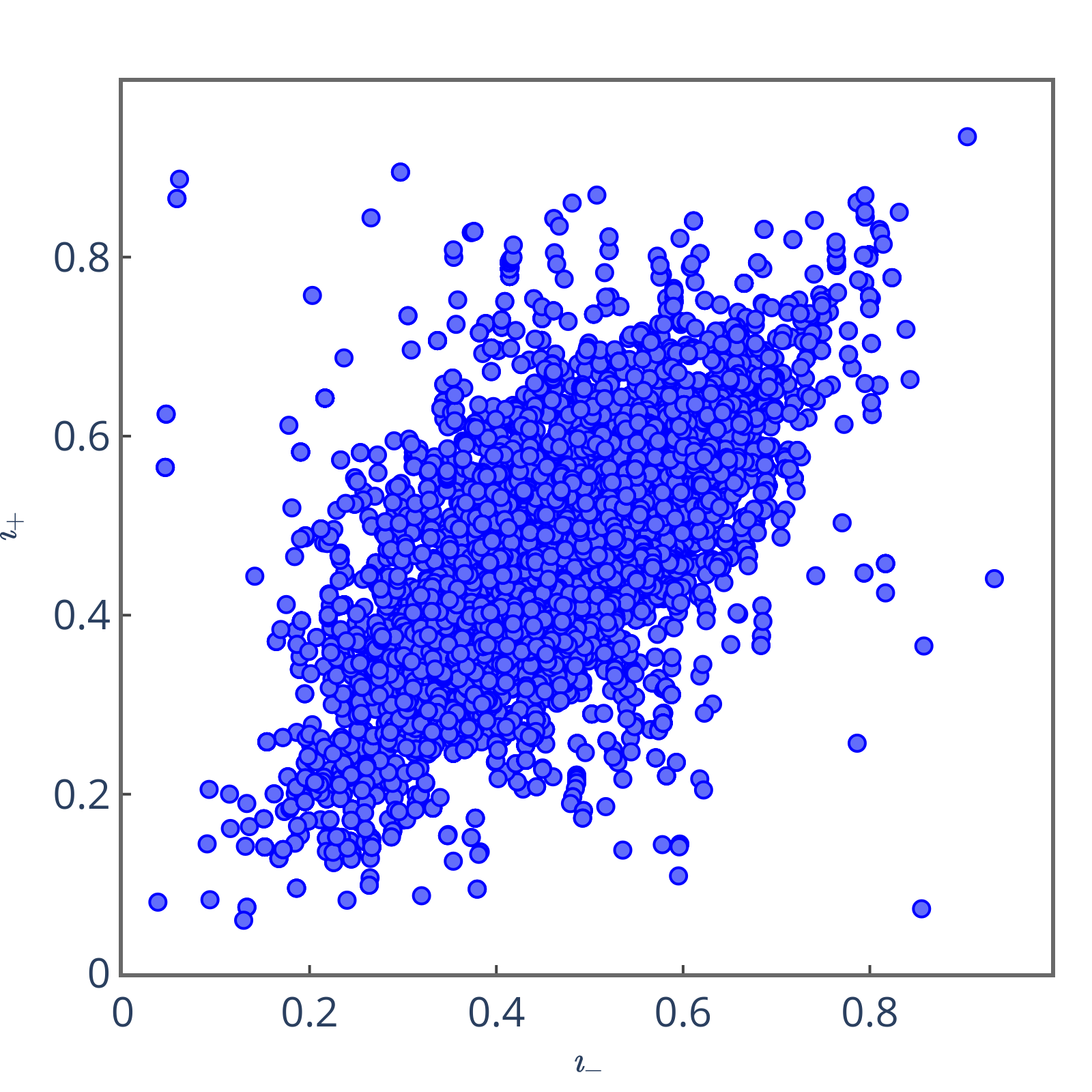}
    \end{minipage}%
    \begin{minipage}[c]{0.5\textwidth}
        \includegraphics[width=\textwidth,keepaspectratio]{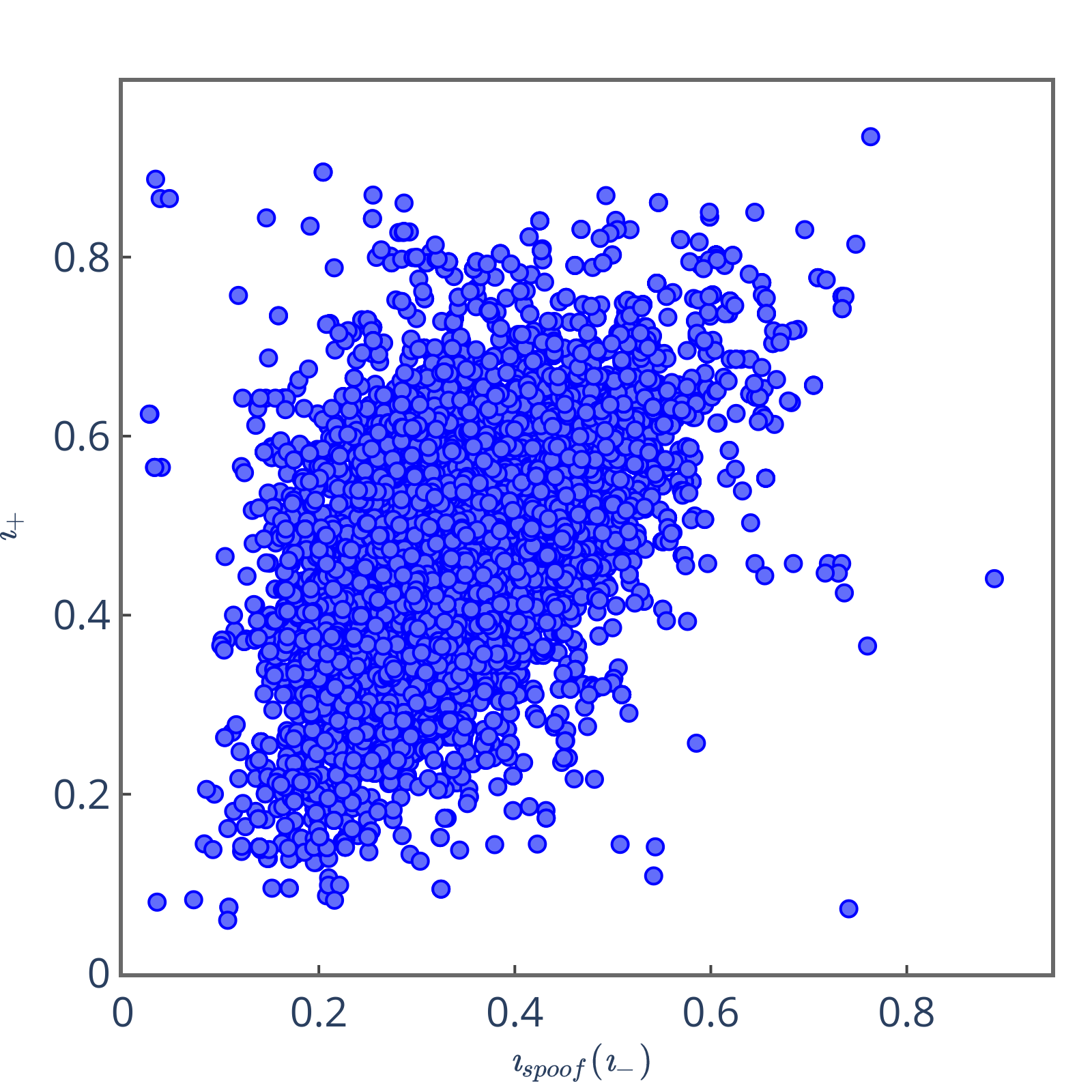}
    \end{minipage}
    \caption{Left panel: Empirical joint distribution of $ \left(\imath_-, \imath_+ \right) $. Right panel: Joint distribution of $\left(\imath_{spoof}(\imath_-), \imath_+ \right)$}
    \label{joint_dis} 
\end{figure}

A possible way to detect spoofing is therefore to compare the long run distribution $(\imath_-, \imath_+)$ with the short term empirical distribution $(\hat{\imath}_-^N, \hat{\imath}_+^N)$.
These two distributions encode the possibility to disentangle legitimate market behavior from spoofed ones.
However, as in the previous approach, the sequence of joint observation is once again not iid.
For short time horizon, the market may be legitimate, though far away from the long run distribution.

\subsection{Monitoring $\hat{\imath}^N_-$ conditioned on $\hat{\imath}^N_+$}
To overcome the previous shortcomings, the next approach is to monitor $\hat{\imath}^N_-(t)$ conditioned on the current market state $\hat{\imath}^N_+(t)$.
From our hypothesis, $\imath_+$ represents the steady state of the market at equilibrium after a market order.
It turns out that conditioned on $\hat{\imath}_+(t)$ the sequence of $\hat{\imath}_-(t)$ is closer to iid.
\begin{figure}[H]
    \centering
    \includegraphics[width=\textwidth,keepaspectratio]{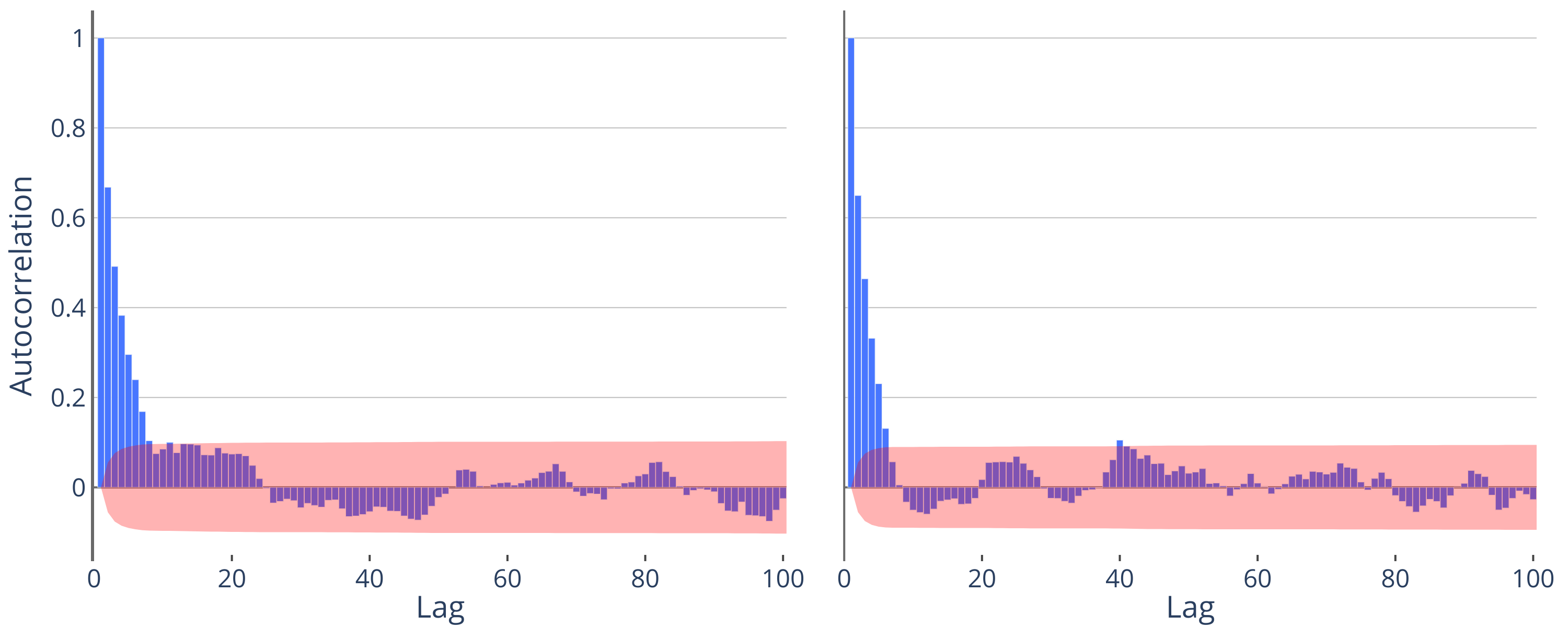}
    \caption{Left panel: Autocorrelation for $0.329\leq \hat{\imath}^N_- | \hat{\imath}^N_+ \leq 0.378$ of stock BMO on June 7, 2017. Right  panel: Autocorrelation for $0.561 \leq \hat{\imath}^N_- | \hat{\imath}^N_+ \leq 0.61$ of stock BMO on June 7, 2017. Red area is the 95\% confidence interval of the autocorrelation.} 
    \label{Conditional_utocorrelation}
\end{figure}
In order to detect spoofing behavior, instead of adopting a statistical test for which some parametric assumptions on the distribution has to be made, we measure the distance between $\imath_-$ and $\hat{\imath}_-^N$ conditioned on the current observed imbalance $\hat{\imath}_+^N$ using a non-parametric distance, the Wasserstein distance, see Appendix for precise definition of this distance.

On the one hand, we know the conditional distribution $\imath_- | \imath_+$ as well as $\imath_{spoof}(\imath_-)$.
Hence, we can deduce the conditional distribution of $\imath_{spoof} | \imath_+$.
We monitor the following two quantities
\begin{equation*}
    \underbrace{
        t \longmapsto d\left( \hat{\imath}^N_-(t), \imath_- |\imath^N_+(t) \right)}
    _{
        \substack{
            \text{Distance from the short term imbalance }\hat{\imath}_-^N(t)\\
            \text{to the \textbf{legitimate} imbalance }\imath_-\\
            \text{given that }\imath_+\sim \imath^N_+(t)
        }
    }
    \quad \text{and} \quad
    \underbrace{
    t \longmapsto d\left( \hat{\imath}^N_-(t), \imath_{spoof}|\imath_+^N(t)\right)}
    _{
        \substack{
            \text{Distance from the short term imbalance }\hat{\imath}_-^N(t)\\
            \text{to the \textbf{spoofed} imbalance }\imath_{spoof}\\
            \text{given that }\imath_+\sim \imath^N_+(t)
        }
    }
\end{equation*}

In Figure \ref{fig:wasserstein_1} and Figure \ref{fig:wasserstein_2}, are the plots thereof for the selected stocks.
We mark in red the area where the distance of the short term imbalance to the spoofed one is smaller than the distance to the legitimate one.
Since the spoofed imbalance is computed regardless whether spoofing is rewarding or not, we additionally mark in blue the area where spoofing is not worth according to Proposition \ref{prop:existence_spoofing}, that is when
\begin{equation}\label{eq:no-spoofing}
    \sup_k \left\{2 \bar{\rho}\mu^+(1 -\bar{\imath})\bar{\imath} w_k - Q_k\bar{\rho}- \nu_k  \right\} \leq 0
\end{equation}
where $\bar{\rho}$ represent the mean of the total volume of market orders purchased divided by the total volume available within $N$ ticks on the limit order book, $\bar{\imath}$ is the mean of $\imath^N_+(t)$.

According to Table \ref{table:stats_stocks}, different stocks have different frequencies, depth of order book as well as volume of incoming market orders.
Our approach does take these different factors into account and are reflected into the different plots.
According to the previous results, spoofing is more likely to happen if either $\mu^+$ -- the overall price impact -- or if $\rho$ -- the relative size of the market orders with respect to the liquidity present in the limit order book -- is large.
We recap for the stocks under study some of their key aspects as well as $\bar{\rho}$, the total amount of market orders within the observed time window relative to the average liquidity available within the observed depth.
\begin{table}[H]
    \begin{center}
            \begin{tabular}{@{}lcrrrr@{}}
                \toprule
                \multicolumn{1}{c}{Stock} & & \multicolumn{1}{c}{$f$ } & \multicolumn{1}{c}{Depth}  & \multicolumn{1}{c}{$\bar{\rho}$} & \multicolumn{1}{c}{$\mu^+$}\\
                                       \midrule
                AEM && 4   & 4 &54.1\% &0.411  \\
                BB  && 38  & 4 &17.8\% &0.467  \\
                BMO && 11  & 4 &62.9\% &0.103   \\
                CNR && 6   & 4 &89.5\% &0.398\\
                CPG && 53  & 5 &26.4\% &0.100   \\
                FNV && 3   & 5 &81.8\% &0.404  \\
                FR  && 60  & 3 &21.4\% &0.209 \\
                PPL && 26  & 4 &54.1\% &0.076 \\
                TD  && 20  & 4 &48.2\% &0.118  \\
                VET && 6   & 5 &75.4\% &0.315 \\
                \bottomrule
            \end{tabular}
            \caption{Frequency, depth, market order volume relative to liquidity available $\bar{\rho}$, price impact $\mu^+$ for the studied Stocks.}
    \end{center}
\end{table} 

%
%

Overall, we do not observe many crossings, if ever.
If a significant crossing happens, it is isolated showing some abnormal behavior.
This is particularly obvious for CPG and PPL.

An outlier in this series of observations is the stock TD where many crossings happens.
TD is a stock which is particularly active with a high rate of market orders -- the largest in our study group.
Furthermore the price impact is particularly low $0.118$.
However, from these market specificities, according to the equation \eqref{eq:no-spoofing}, most of the time it is not worth spoofing.
Hence, there remain only one significant crossing out of the area where spoofing would be worthwhile.

\begin{figure}[H]
    \centering
    \includegraphics[width=\textwidth,keepaspectratio]{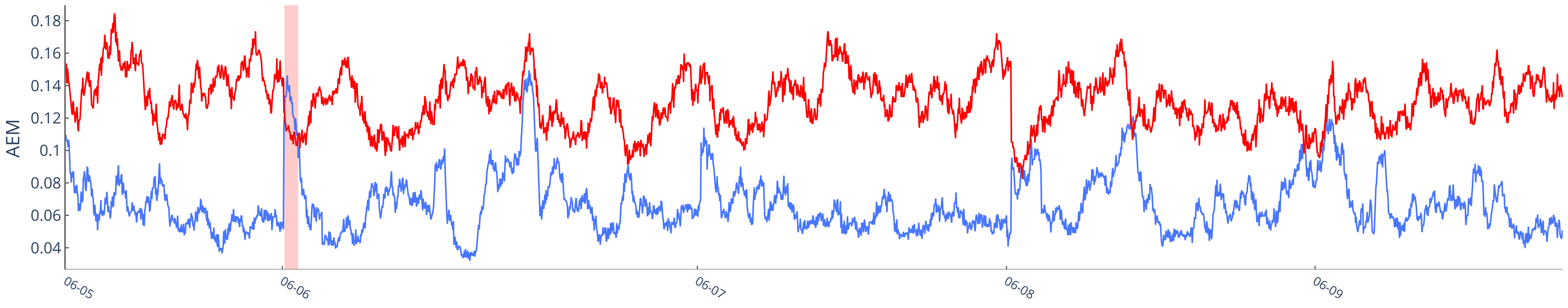}
    \includegraphics[width=\textwidth,keepaspectratio]{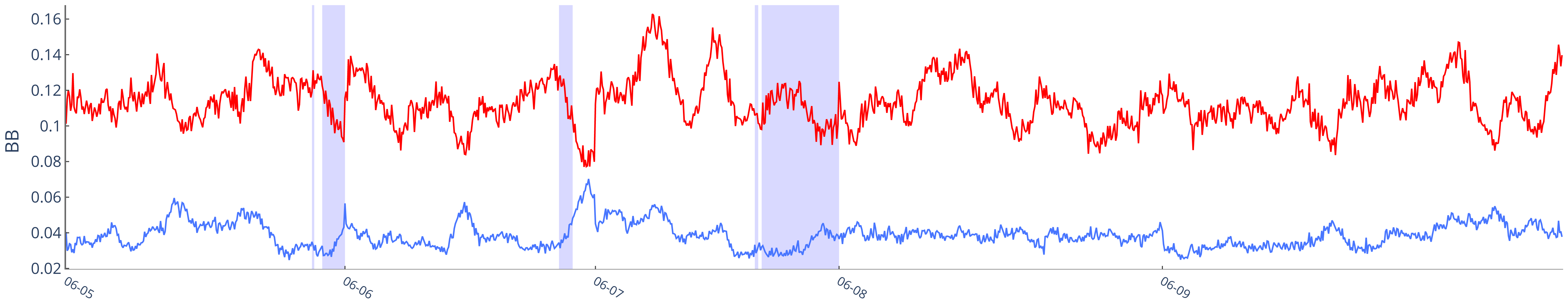}
    \includegraphics[width=\textwidth,keepaspectratio]{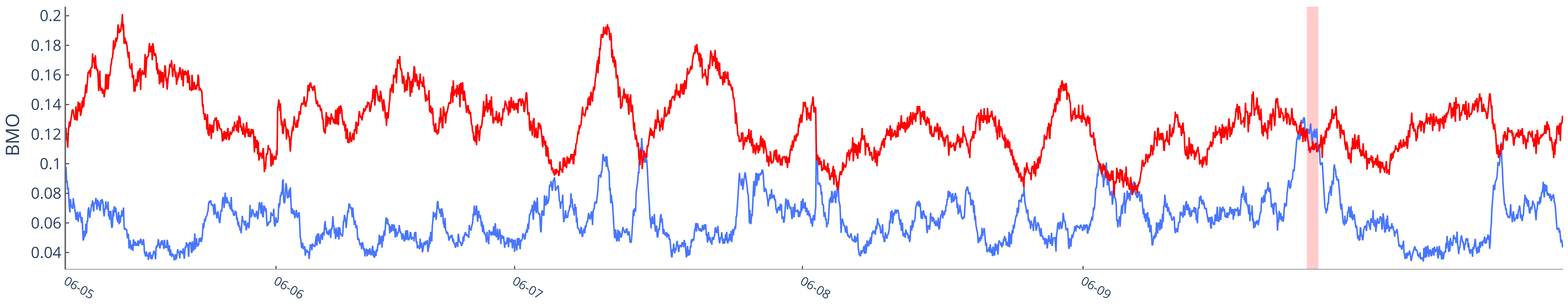}
    \includegraphics[width=\textwidth,keepaspectratio]{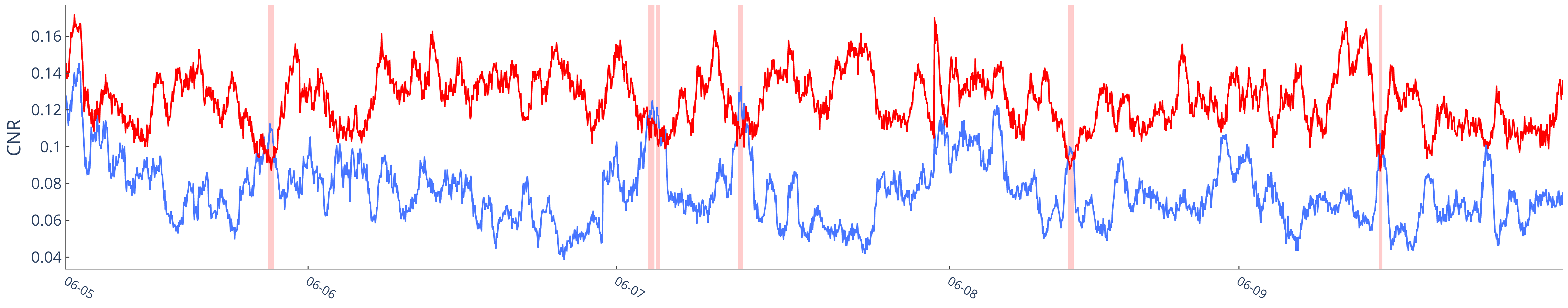}
    \includegraphics[width=\textwidth,keepaspectratio]{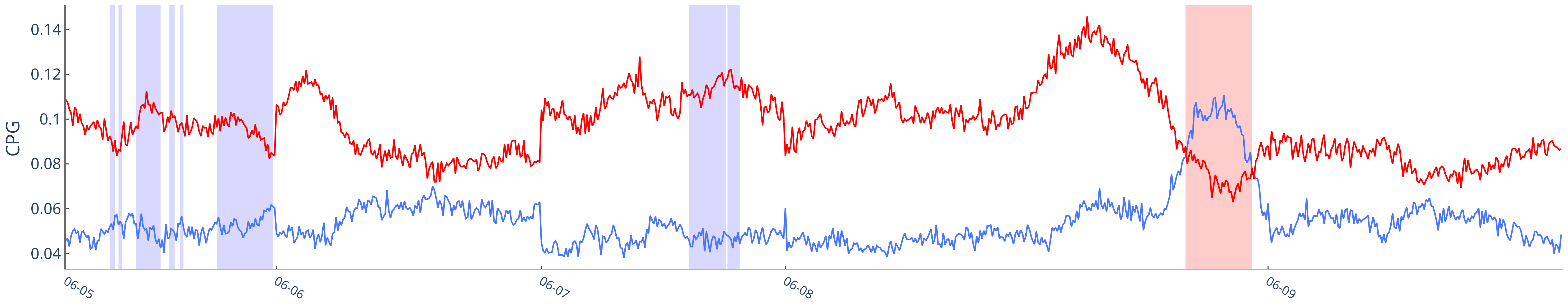}
    \caption{Time series of $t \mapsto d\left(\imath_{-}, \hat{\imath}^N_-(t) |\hat{\imath}^N_+(t) \right)$ (blue line) and $t_k \mapsto d\left(\imath_{spoof}, \hat{\imath}^N_-(t) |\hat{\imath}^N_+(t) \right)$ (red line) from June 5, 2017 to June 9, 2017. Red area is where $d\left(\imath_{spoof}, \hat{\imath}^N_- |\hat{\imath}^N_+ \right)\leq d\left(\imath_-, \hat{\imath}^N_- |\hat{\imath}^N_+ \right)$ over more than 10 consecutive times. Blue area  is where $\sup_k \left\{2 \bar{\rho}\mu^+(1 -\bar{\imath})\bar{\imath} w_k - Q_k\bar{\rho}- \nu_k  \right\} \leq 0 $  and $\bar{\rho}$, $\bar{\imath}$ are the mean of $\rho$, $\imath_N^+$ in each window respectively. From top to bottom are stock AEM, BB, BMO, CNR, and CPG respectively.}
    \label{fig:wasserstein_1}
\end{figure}

\begin{figure}[H]
    \centering
    \includegraphics[width=\textwidth,keepaspectratio]{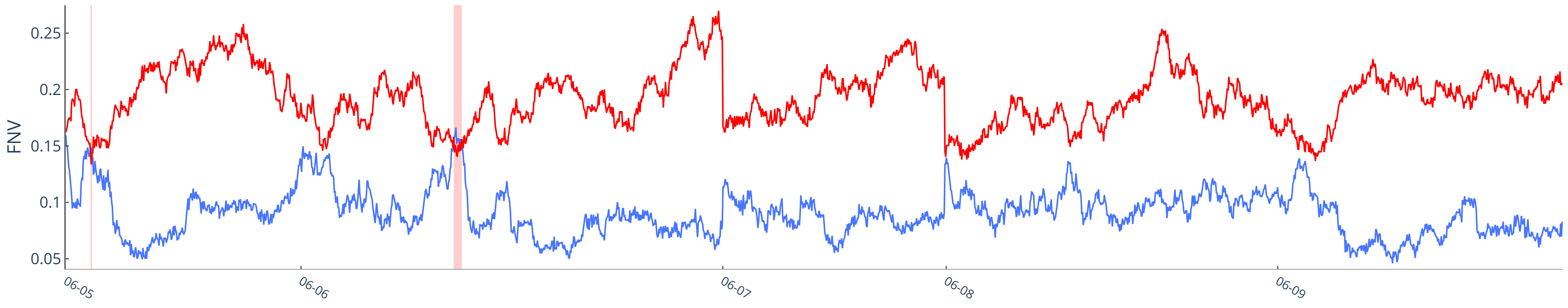}
    \includegraphics[width=\textwidth,keepaspectratio]{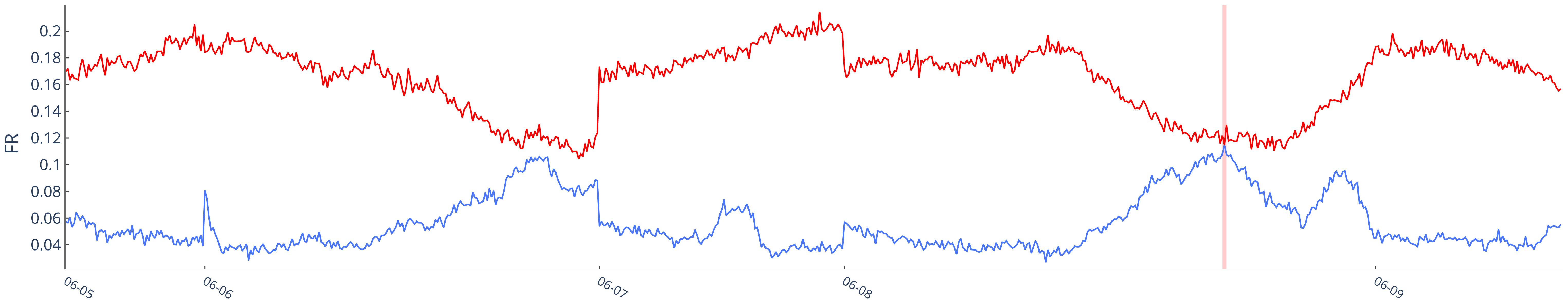}
    \includegraphics[width=\textwidth,keepaspectratio]{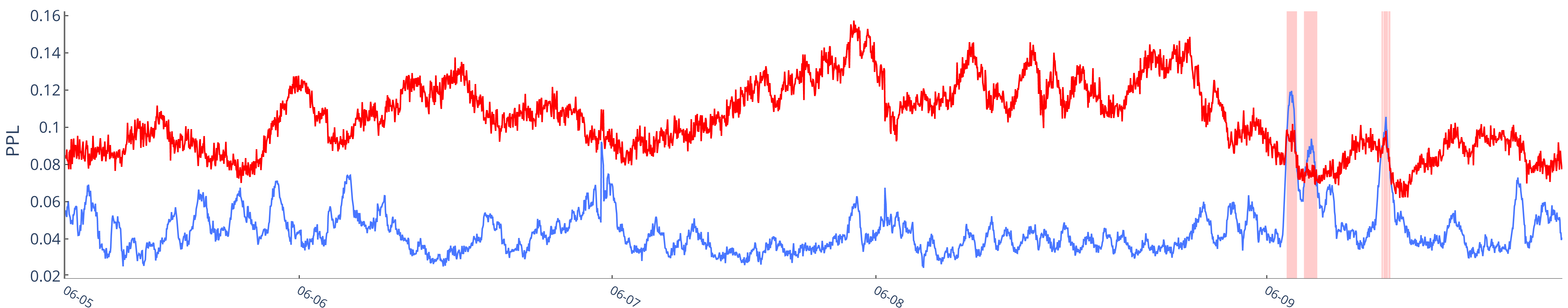}
    \includegraphics[width=\textwidth,keepaspectratio]{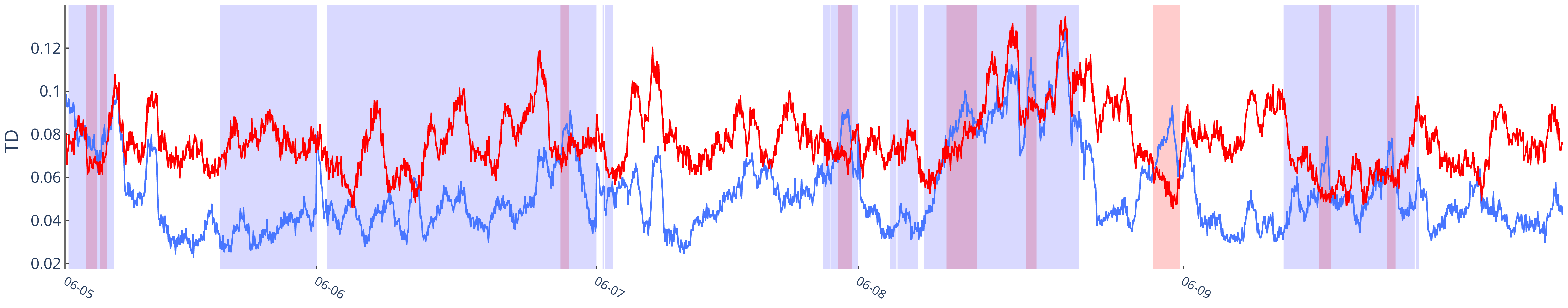}
    \includegraphics[width=\textwidth,keepaspectratio]{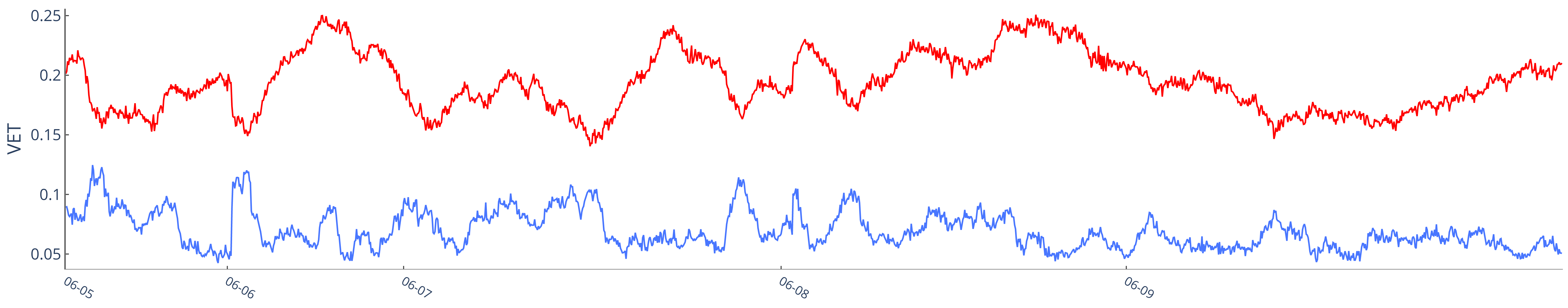}
    \caption{Time series of $t_k \mapsto d\left(\imath_{-}, \hat{\imath}^N_-(t) |\hat{\imath}^N_+(t) \right)$ (blue line) and $t_k \mapsto d\left(\imath_{spoof}, \hat{\imath}^N_-(t) |\hat{\imath}^N_+(t) \right)$ (red line) from June 5, 2017 to June 9, 2017. Red area is where $d\left(\imath_{spoof}, \hat{\imath}^N_- |\hat{\imath}^N_+ \right)\leq d\left(\imath_-, \hat{\imath}^N_- |\hat{\imath}^N_+ \right)$ over more than 10 consecutive times. Blue area  is where $\sup_k \left\{2 \bar{\rho}\mu^+(1 -\bar{\imath})\bar{\imath} w_k - Q_k\bar{\rho}- \nu_k  \right\} \leq 0 $  and $\bar{\rho}$, $\bar{\imath}$ are the mean of $\rho$, $\imath_N^+$ in each window respectively. From top to bottom are stock FNV, FR, PPL, TD and VET respectively.
    }
    \label{fig:wasserstein_2}
\end{figure}

\section{Conclusion}
In this paper we address the question of assessing quantitatively eventual spoofing behavior in high frequency trading.
In a stylised setting we present how a spoofing strategy from a market taker or maker is designed by manipulating the imbalance at different depth level to impact the subsequent price movement.
We provide and discuss the conditions for the market to allow for spoofing manipulations.
We subsequently solve the optimization problem from a spoofer perspective and derive/discuss the resulting imbalance after spoofing as a function of the market parameters.
We calibrate the weighted imbalance and price movement impact to Level 2 data provided by TMX.
Using these results we propose a quantification instrument to monitor in real time eventual spoofing behavior on the market using a conditional Wasserstein distance.
We illustrate these results on the data provided by TMX.

This approach is by no means a definitive answer to spoofing detection but rather a first take on.
The dynamic structure of the limit order book and strategy, the memory dependence of the parameters over time, as well as the specificities of one market with respect to another one are left to further study.
Also left to further studies is the consideration of multiple venue that could also be integrated into this framework.
Furthermore, there might be alternative approaches subject to new research directions -- monitoring arrival rates of orders, frequency of book/cancelling, etc. -- that could complement such a monitoring approach.

\section{Credits}
We thank Fields Institute for the organization of the series ``Fields-China Joint Industrial Problem Solving Workshop'' through which we got involved into this problem.
Thanks to TMX and their data analysis team for providing unique access to large datasets, high performance computing facilities as well as precious insights concerning high frequency trading.
Finally, the data analysis could not have been performed without the outstanding work and dedication of the open source community for the development of scientific computing/visualisation libraries/platforms such as NumPy \citep{numpy2020}, SciPy \citep{scipy2020}, Mistic \citep{mistic2011}, Pandas \citep{reback2020pandas, pandas2010}, Apache Spark, Plotly to name the most relevant for this work.

\appendix

\section{Proofs}\label{appendix:proofs}
\begin{proof}[Proof of Proposition \ref{prop:existence_spoofing}]
    Let $H = \rho a$, and the imbalance $\imath_k(v) = b/(b+a + w_k v) $ with $b = a \bar{\imath}/(1-\bar{\imath})$.
    It follows that the gradient of $\imath_k(v)$ is given by
    \begin{equation*}
        \nabla \imath_k(v) = - \frac{b}{\left( a+b +  w_k v \right)^2} w_k = - \frac{(1-\bar{\imath})}{a\bar{\imath}} \imath^2_k(v) w_k
    \end{equation*}
    \begin{itemize}[fullwidth]
        \item If $\bar{\imath}\leq 1/2$, from the previous equations, since $g(x) = x^2/(2a)$, it follows that there is no spoofing manipulation if and only if
            \begin{equation*}
                f(v):=Q_k\frac{v^2}{2a} +  Q_k \rho v + v \nu_k - 2\rho a \mu^+\left( \bar{\imath} - \imath_k(v)  \right) \geq 0
            \end{equation*}
            for any $v \geq 0$.
            Taking the gradient for this function yields
            \begin{equation*}
                \nabla f(v) = Q_k\frac{v}{a} + Q_k \rho + \nu_k -2 \rho \mu^+ \frac{1-\bar{\imath}}{\bar{\imath}} \imath^2_k(v) w_k
            \end{equation*}
            which is a monotone functional in $v$.
            Since $f(0) =0$, it follows that $f (v) \geq 0$ for any $v$ if and only if $\nabla f(0) \geq 0$ which is equivalent to
            \begin{equation*}
                Q_k  \rho + \nu_k \geq 2 \rho  \mu^+ (1-\bar{\imath})\bar{\imath} w_k
            \end{equation*}
        \item If $\bar{\imath}>1/2$, there is no spoofing manipulation if and only if 
            \begin{equation*}
                f(v):=q_k\frac{v^2}{2a} +  \rho q_k v - q_kkv - 2\rho a \mu^+\left( 1/2 - \imath_k(v)  \right) \geq 0
            \end{equation*}
            for any $v \geq 0$ the gradient of which is given by
            \begin{equation*}
                \nabla f(v) = q_k\frac{v}{a} + \rho q_k - q_kkv  -2 \rho \mu^+ \frac{1-\bar{\imath}}{ \bar{\imath}} \imath^2_k(v) w_k
            \end{equation*}
            Since $f(0)>0$, as previously argued, it follows that $f(v) \geq 0$ as soon as $\nabla f(0)\geqslant 0$, which yields the same conditions.
    \end{itemize}
\end{proof}

\begin{proof}[Proof of Proposition \ref{prop:optimal_spoofing}]
    Adopting the notations $Q:=Q_k$, $\nu = \nu_k$, $w = w_k$, $H = \rho a$, the goal is to optimize over $v\geq 0$ the objective function 
    \begin{align*}
        f(v) & = (1-Q) \frac{(\rho a)^2}{2a} + Q\frac{(\rho a +v)^2}{2a} + \rho a \mu^+ \left( 2\imath(v) - 1 \right)+ v \nu\\
             & = \frac{(\rho a)^2}{2a} + Q \frac{v^2}{2a} + \left( Q \rho + \nu \right)v +\rho a \mu^+ \left( 2\imath(v) - 1 \right)
    \end{align*}
    First order condition with Lagrangian $\lambda$ yields
    \begin{equation*}
        Q \frac{v}{a} + \left( Q\rho +\nu \right) - 2\rho w \mu^+ \frac{1-\bar{\imath}}{\bar{\imath}} \imath^2 = \lambda 
    \end{equation*}
    where $\imath : = \imath(v)$.
    Solving as a function of $\imath$ in $(0,1)$, we get
    \begin{align*}
        \lambda(\imath) & = \left[\left( Q\rho +\nu \right) - 2\rho w\mu^+ \frac{1-\bar{\imath}}{\bar{\imath}} \imath^2 \right]^+ \\
        v(\imath) & = \frac{a}{Q}\left[ 2\rho w\mu^+ \frac{1-\bar{\imath}}{\bar{\imath}} \imath^2  - \left( Q\rho +\nu \right) \right]^+
    \end{align*}
    Given now the optimal $v(\imath)$ as a function of $\imath$, we solve for $\imath$ such that
    \begin{equation*}
        \frac{1}{\imath} = \frac{a + b + w v(\imath)}{b} = \frac{1}{\bar{\imath}} + w\frac{1-\bar{\imath}}{a\bar{\imath}}v\left( \imath \right) = \frac{1}{\bar{ \imath}} + \frac{w}{Q} \frac{1-\bar{\imath}}{\bar{\imath}} \left[ 2\rho w\mu^+ \frac{1-\bar{\imath}}{\bar{\imath}} \imath^2  - \left( Q\rho +\nu \right) \right]^+
    \end{equation*}
    Since the left hand side in strictly decreasing on from $\infty$ to $\frac{1}{\bar{\imath}}$ on $(0, \bar{\imath}]$ and the right hand side is increasing from $\frac{1}{\bar{\imath}}$, on $(0, \bar{\imath}]$, there exists a unique solution which is a cubic root.
\end{proof}

\section{Round Trip Situation}\label{appendix:round_trip}

In this paper we mainly focus on the spoofing behavior from a market taker's viewpoint.
As for a market maker, spoofing behavior might be rewarding as well.
However, as seen in the following subsection, the rewards from spoofing are intertwined with the ones from pure market making.

We present a simple situation together with the numerical analysis in a blocked shape setting with the same model assumptions as before.
We assume that the potential market maker spoofer acts as follows:
At the first stage it decides to spoof with a volume $v$ at depth $k$ on the ask side to drive the price down and acquire an amount $H$ of shares after this price movement.
When the market comes back to its steady state, it liquidates $H$ and eventually buys back $v$ if it has been executed.
We assume that $v$ and $H$ are decided at the very beginning.\footnote{
    This stylised situation makes strong assumptions and simplifications.
    First $H$ is decided at time $0$ even if it is executed after the price movement.
    This is to prevent conditional optimization.
    Second, the liquidation of the inventory $H$ and $v$ occurs separately.
    Once again, to provide simplified optimization problem, while we could numerically consider a liquidation of the net inventory $H-v$.
    Finally, a second spoofing could happen at the second stage as in the previous section to liquidate the inventory.
}

After spoofing a volume $v$ at level $k$, as soon as the price moves the spoofer executes its market order $H$ for a revenue of
\begin{equation*}
-\left( H - v 1_{\{y>k\}} \right)\left( p + \delta \Delta + \delta (x+y) \right) - \delta G_{a}(H) - \delta (y-k) v 1_{\{y>k\}}
\end{equation*}
where $G_a(H)=H^2/(2a)$, $p = (p^+ +p^-)/2$, and $\Delta = (p^+ - p^-)/(2\delta)$ is the effective spread in ticks.
The spoofer then waits for the market to return to its steady state and liquidate the resulting inventory with market orders.
For ease of computation, we assume that it executes two market orders: One for $H$ and one for $v$ if it has been executed\footnote{Combining both in terms of $H-v1_{\{y>k\}}$ is cost effective but complicates the exposition of the result.} for a revenue of
\begin{equation*}
    H p - \delta \Delta H - G_b(H) -1_{\{y>k\}} \left( v (p + \delta \Delta) + \delta G_a(v) \right)
\end{equation*}
Adding both and integrating yields an average net revenue of
\begin{multline*}
    \frac{R(H, v)}{\delta} = - H \left(2\Delta + \mu^+\left( 2\imath_k(v) -1 \right)\right) - G_a(H) - G_b(H) \\
    + Q_k\left[ \left(k+\mu^+\left( 2\imath_k(v) - 1 \right)\right) v - G_a(v) \right]\\
    = - H \left(2\Delta + \mu^+\left( 2\imath_k(v) -1 \right)\right) - \frac{1}{a \bar{\imath}}\frac{H^2}{2} + Q_kv\left[ k + \mu^+\left( 2 \imath_k(v) - 1 \right) \right] - \frac{Q_k}{a} \frac{v^2}{2}
\end{multline*}

From this equation, we can derive the following remarks concerning the decision of the spoofer:
\begin{itemize}
    \item If $v=0$: This corresponds to the classical situation where a market maker takes advantage of the temporary market movement to execute a market order and cash out at a later time when the market comes back to its steady state.
        Clearly, it gets a positive gain if and only if
        \begin{equation*}
            \bar{\imath} \leq \frac{1}{2} - \frac{\Delta}{\mu^+}
        \end{equation*}
        In particular, if the effective spread $\Delta$ is large, or if $\mu^+$ is small, then it is impossible or the initial imbalance should be very small.
        In the case where this happens, then $\bar{H}^\ast$ is given by
        \begin{equation*}
            \bar{H}^\ast = a \bar{\imath}\left( 2 \Delta + \mu^+\left( 2\bar{\imath} -1 \right) \right)^-
        \end{equation*}
        with corresponding revenue of
        \begin{equation*}
            \bar{R}^\ast = \frac{1}{2}\left[a \bar{\imath}\left( 2 \Delta + \mu^+\left( 2\bar{\imath} -1 \right) \right)^-\right]^2
        \end{equation*}

    \item If $H = 0$: This corresponds to the classical situation where a market maker posts limit orders at a given depth to gain from possible fluctuations.
        This results in corresponding average revenue given $v$ of 
        \begin{equation*}
            \hat{R}(v) = Q_k v\left(k + \mu^+\left( 2\imath_k(v)-1 \right)\right) - \frac{Q_k}{a} \frac{v^2}{2} 
        \end{equation*}
        From this equation, even if the spoofer gets a positive gain of $k$ ticks buy executing its order, it will drive the imbalance $\imath_k(v)$ below $1/2$ and face adverse price movement that will offset its gains.
        The optimal $\hat{v}^\ast = \hat{v}^\ast(\bar{\imath})$ in that situation is not explicit, but can be easily numerically implemented and corresponds to an optimal revenue of
        \begin{equation*}
            \hat{R}^\ast = Q_k \hat{v}^\ast\left(k + \mu^+\left( 2\imath_k(\hat{v}^\ast)-1 \right)\right) - \frac{Q_k}{a} \frac{\hat{v}^2}{2}
        \end{equation*}
\end{itemize}
In general, solving for the optimal $H$ is straightforward with
\begin{equation*}
    H^\ast = a \bar{\imath}\left( 2 \Delta + \mu^+\left( 2 \imath_k(v)-1 \right)\right)^-
\end{equation*}
and corresponding average revenue:
\begin{multline*}
    \frac{R(v)}{\delta} = \frac{1}{2}\left[a \bar{\imath}\left( 2 \Delta + \mu^+\left( 2\imath_k(v) -1 \right) \right)^-\right]^2 + Q_k v\left(k + \mu^+\left( 2\imath_k(v)-1 \right)\right) - \frac{Q_k}{a} \frac{v^2}{2}
\end{multline*}
These two effects are difficult to disentangle from a truly spoofing gain when $H$ as well as $v$ are strictly positive.
However, this can be done numerically and the results are presented in Figure \ref{fig:round_trip}, where the spoofing region -- $H>0$ as well as $v>0$ -- is indicated.

We can however draw some stylised facts about the spoofing behavior from this market maker viewpoint.
The impact of the different parameters -- initial imbalance $\bar{\imath}$, probability of getting executed $Q_k$, local sensitivity of imbalance on the price impact $w=w_k$ as well as overall price deviation $\mu^+$ are similar to the previous case.
However, in addition to the previous part, the effective spread $\Delta$ acts negatively on the spoofing opportunity in that context.
Indeed, a positive market order $H$ is only triggered if $\mu^+(2 \imath^\ast - 1) \leq -2 \Delta$, which requires a spoofed imbalance satisfying
\begin{equation*}
    \imath_{spoof} \leq - \frac{\Delta}{\mu^+} + \frac{1}{2}
\end{equation*}
If $\Delta$ is too large or $\mu^+$ too low, a spoofing strategy is no longer rewarding.

\begin{figure}[H]
    \centering
    \includegraphics[width=0.8\textwidth,keepaspectratio]{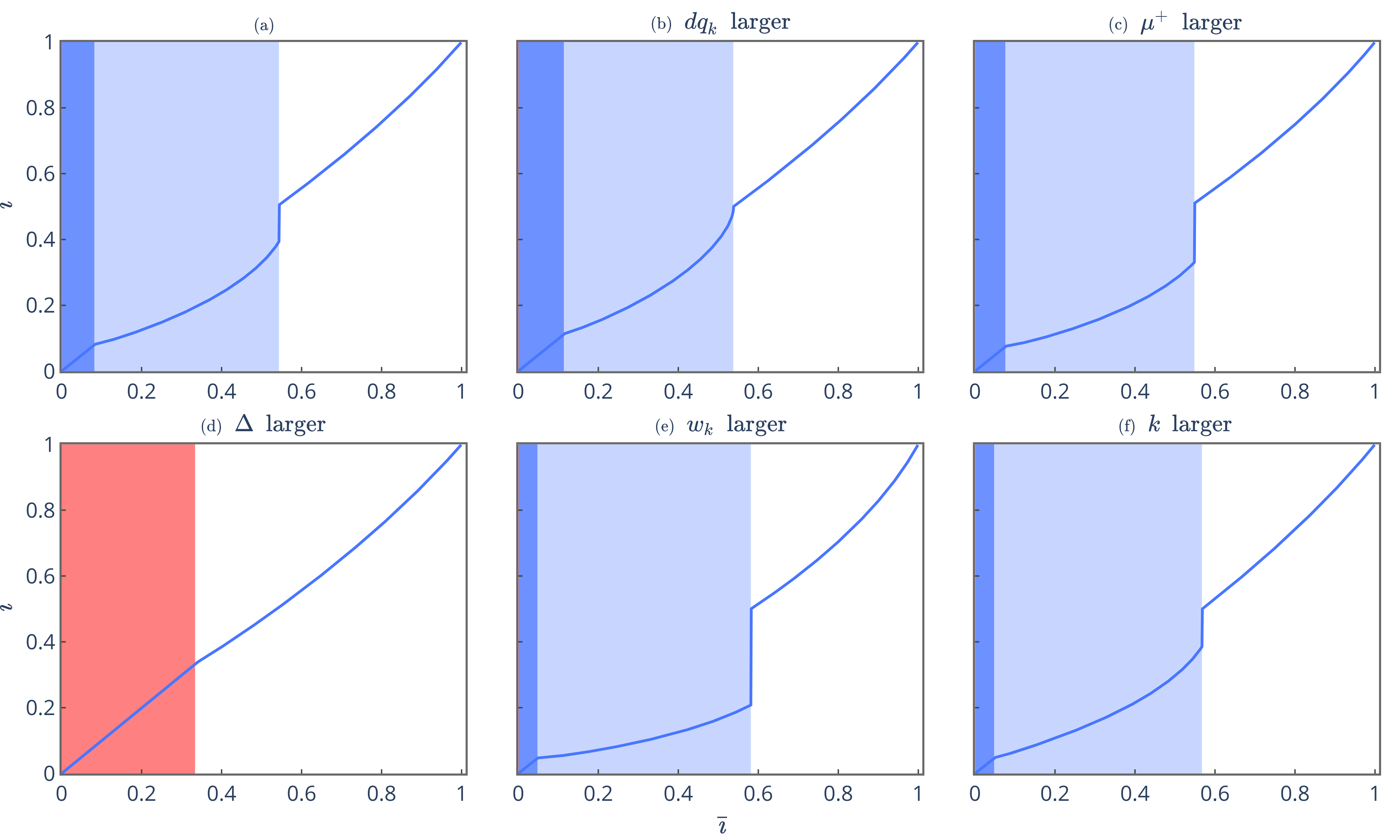}
    \caption{$\imath_{spoof}$ as a function of $\bar{\imath}$ and in $(a)$
    $\mu^ + =3,  \Delta = 0, k=1, w_k =0.1,  dq_y = 0.001 $ for all $y \geq k $.
    One parameter is increased each time with respect to $(a)$ where 
    $(b)$: $dq_y = 0.002 $ for all $y \geq k $; 
    $(c)$: $\mu^+ = 4$;
    $(d)$: $\Delta = 2$;
    $(e)$: $w_k = 0.4$;
    $(f)$: $k = 4$.
    Red area: $v^\ast=0, H^\ast=0$; Dark blue area: $v^\ast=0, H^\ast>0$; Light blue area: $v^\ast>0, H^\ast>0$; White area:  $v^\ast>0, H^\ast=0$.}
    \label{fig:round_trip}
\end{figure}

\section{Goodness of Fit}\label{appendix:goodness_fit}
The model for the price movement is conditioned on the imbalance.
So, as for the goodness of fit, we compare the empirical price change distribution with the fitted one conditioned on different level of imbalance.
\begin{figure}[H]
    \centering
    \includegraphics[width=\textwidth,keepaspectratio]{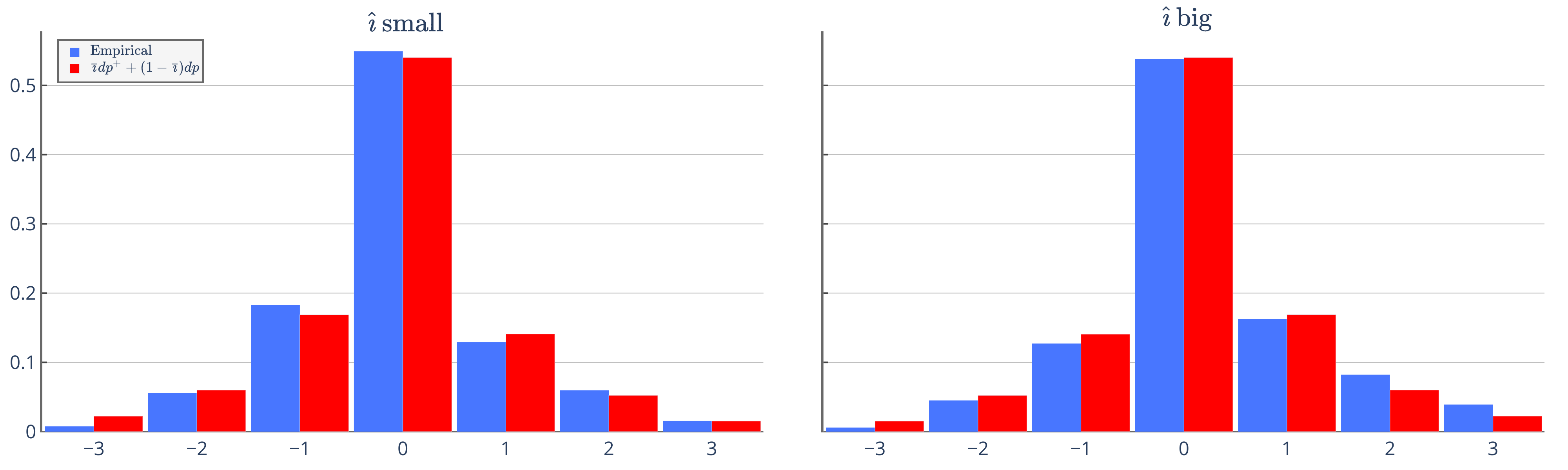}
    \caption{Empirical price change conditioned on the imbalance (blue) and $\bar{\imath} dp^+ +(1 - \bar{\imath})dp^- $ (red). Left panel, $\hat{\imath} <  0.298$ and right panel, $\hat{\imath} \geq 0.648$, which correspond to the low and top 5\% quantiles.} 
    \label{fig:conditioned_dpa}
\end{figure}
Since we fit a conditional distribution, we bucket the observed imbalance into 20 equidistant quantile intervals and perform a $\mathcal{X}^2$ test of goodness of fit between the empirical distribution conditioned on the imbalance within this bucket and $\bar{\imath} dp^+ + (1-\bar{\imath})dp^-$ where $\bar{\imath}$ is taken as the mid value of the imbalance for the corresponding bucket.
For each bucket, we have about 500 data points which we test against the null hypothesis that the empirical price change distribution is consistent with $\bar{\imath} dp^+ +(1 - \bar{\imath})dp^- $. 
Table \ref{table:p_value} shows the $\mathcal{X}^2$-statistics and $p$-values in each bucket.\footnote{Significance level $0.05$ corresponds to a $\mathcal{X}^2_6(0.95) = 12.592$.}
3 buckets exceed the bound (marked in black) where the null hypothesis is rejected while for the 17 others, the null hypothesis sustains.
\begin{table}[H]
    \begin{center}
        \resizebox{\columnwidth}{!}{
        \begin{tabular}{@{}ccccccccccc@{}}
            \toprule
            $l$ & 1 & 2& 3& 4& 5& 6& 7& 8& 9 &10 \\
            \hline
            $\bar{\imath}_l$& 0.245& 0.321& 0.356& 0.381& 0.4& 0.416& 0.43& 0.443& 0.456& 0.468            \\
            $\mathcal{X}^2$-statistic & 7.112& 7.635& 12.476& 6.253& 6.658& 8.041& 4.089& 2.814& 12.56& 4.177\\
            $p$-value & 0.311& 0.266& 0.052& 0.395& 0.354& 0.235& 0.665& 0.832& 0.051& 0.653\\
            \midrule
            $l$ & 11& 12& 13& 14& 15& \textbf{16}& 17& \textbf{18}& 19 & \textbf{20} \\
            \hline
            $\bar{\imath}_l$& 0.48& 0.493& 0.506& 0.52& 0.534& \textbf{0.549}& 0.568& \textbf{0.591}& 0.624& \textbf{0.702}            \\
            $\mathcal{X}^2$-statistic & 9.372& 8.775& 5.219& 3.014& 7.634& \textbf{16.046}& 2.334& \textbf{17.835} & 6.949& \textbf{17.953}\\
            $p$-value  & 0.154& 0.187& 0.516 & 0.807& 0.266& \textbf{0.014}& 0.887& \textbf{0.007}& 0.326& \textbf{0.006}\\
            \bottomrule
        \end{tabular}
    }
        \caption{Chisquare test on stock BMO from June 5, 2017 to June 9, 2017. }
        \label{table:p_value}
    \end{center}
\end{table}

\section{Computation of $\imath_{spoof}$ }\label{appendix:computation_ispoof}
For $t$ in $\Pi$, we compute $\hat{\imath}_-(t)$ and $\hat{\imath}_+(t)$, the imbalance before and after a market order, as follows:
\begin{align*}
    \hat{\imath}_-(t) &= \frac{\sum\limits_{k\leq N} \sum\limits_{t - f  \leq s<t} w_k \bar{v}^-_k(s)}{\sum\limits_{k\leq N} \sum\limits_{t - f  \leq s<t}  w_k \left(\bar{v}^-_k(s) + \bar{v}_k(s) \right)}\\
    \hat{\imath}_+(t) &= \frac{ \sum\limits_{k\leq N} \sum\limits_{t  \leq s<t+1} w_k \bar{v}^-_k(s)}{\sum\limits_{k\leq N} \sum\limits_{t\leq s<t+1}  w_k \left(\bar{v}^-_k(s) + \bar{v}_k(s) \right)} 
\end{align*}

For each $\hat{\imath}_-(t)$, the optimal spoofing strategy $v_{spoof}$ can be solved explicitly from
\begin{equation*}
    v_{spoof,k} = 1+\frac{ \left(1-\hat{\imath}_-(t) \right) w_k}{Q}\left[2\rho_t w_k \mu^+ \frac{1-\hat{\imath}_-(t)}{\hat{\imath}_-(t)} \imath^2 - \left( Q_k \rho_t +\nu \right)\right]^+ 
\end{equation*}
where $\imath = \frac{b_t}{b_t + a_t + w_k v_{spoof,k} }$, $a_t$ is the average size of the limit order book $f$ seconds before a market order
\begin{align*}
    a_t &= \frac{\sum\limits_{k = 1}^N \sum\limits_{t -f\leq s< t} \bar{v}_k(s) \Delta s}{Nf} \\
    \rho_t &= \frac{\sum\limits_{k = 1}^N \sum\limits_{t -f\leq s< t} H_s }{a_t}
\end{align*} 
where $H_s$ is the market order volume at time $s$.
In the same way, we can define $b_t$ and 
\begin{equation*}
    \imath_{spoof}(\hat{\imath}_-(t)) = \frac{b_t}{b_t+a_t + \sum_k w_k v_{spoof,k} }
\end{equation*}

\section{Wasserstein Distance, Kernel Approximation and Conditional Estimation}\label{appendix:wasserstein}
For two distributions $\mu$ and $\nu$ the 2-Wasserstein distance is defined as
\begin{align*}
    W_2\left( \mu, \nu \right) & = \left( \inf \left\{ \int_{}^{}\left( x-y \right)^2 \pi(dx,dy)\colon \pi_1 \sim \mu, \pi_2 \sim \nu  \right\}\right)^{1/2}\\
                               & = \left( \int_{0}^{1}\left( q_{\mu}(\alpha) - q_{\nu}(\alpha) \right)^2 d\alpha  \right)^{1/2}
\end{align*}
From a generic perspective, the conditional distance we consider is as follows:
If we assume that 
\begin{equation*}
    \left(\imath_-, \imath_+\right)\sim K(y, dx)\otimes \mu(dy)
\end{equation*}
where $\mu \sim \imath_+$ and $K(y, \cdot) \sim \imath_-|\imath_+=y$, it follows that
\begin{equation*}
    \left(\imath_{spoof}\left( \imath_- \right), \imath_+\right) \sim K_{spoof}\left( y, dx \right)\otimes \mu(dy) \quad \text{where}\quad K_{spoof}(y, \cdot) = K\left( y,\cdot \right) \circ \imath_{spoof}^{-1}  
\end{equation*}
Hence given $\imath_+ =y$, we have
\begin{equation*}
    W_2\left(K(y, \cdot), K_{spoof}(y, \cdot)\right) = \left( \int_{0}^{1}\left( q_{K(y, \cdot)}(\alpha) - \imath_{spoof}\left( q_{K(y, \cdot)} (\alpha)\right) \right)^2 d\alpha  \right)^{1/2} 
\end{equation*}

Heuristically we wish to monitor the following two quantities
\begin{equation*}
    \underbrace{W_2\left( \imath_-, \hat{\imath}_-^N \right)|\imath_+^N}_{\substack{\text{Distance from the short term imbalance }\hat{\imath}_-^N\\\text{to the \textbf{equilibrium} imbalance }\imath_-\\\text{given that }\imath_+\sim \imath^N_+}} \quad \text{and} \quad \underbrace{W_2\left( \imath_{spoof}(\imath_-), \hat{\imath}_-^N \right)|\imath_+^N}_{\substack{\text{Distance from the short term imbalance }\hat{\imath}_-^N\\\text{to the \textbf{spoofed} imbalance }\imath_{spoof}\\\text{given that }\imath_+\sim \imath^N_+}}
\end{equation*}

From the data, we can calibrate the joint distribution $(\imath_-, \imath_+)$ as well as $(\imath_{spoof}(\imath_-), \imath_+)$.\footnote{The former fits well with a joint normal distribution, while the second one with a skewed normal distribution, see Figure \ref{fit_joint_dis}. Other parametrization could eventually be used too.}
Hence, we have a parametrization of $K(y, dx)$ and $K_{spoof}(y, dx)$ for every $y$.
However for each value $\hat{\imath}_+(l)$ from the discrete distribution $\hat{\imath}_+^N$ we only have a single sample point $\hat{\imath}_-(l)$ at hand.
In order to overcome this problem we bucket the values of $\hat{\imath}_+(l)$ in the sample of $\hat{\imath}^N_+$ into several equidistant quantile intervals to get a Kernel approximation of $\hat{\imath}^N_-$.

The monitoring strategy at a given time $t_k$ in $\Pi$ is given as follows
\begin{enumerate}[label=\arabic* - , fullwidth]
    \item Consider the discrete short term joint distribution $(\hat{\imath}^N_-(t_k), \hat{\imath}_+^N(t_k))$ given by the sample
        \begin{equation*}
            \left( \imath_-(s), \imath_+(s) \right),\quad s = t_k, \ldots, t_{k-N+1}
        \end{equation*}
        of the last $N$ pairs of imbalances before time $t$.
    \item We define the following $L$ buckets of equal cardinality $N/L$ 
        \begin{equation*}
            J_l = \left\{ s \colon s=t_k, \ldots, t_{k-N+1}, q_{\hat{\imath}_+^N}\left( \frac{l-1}{L} \right)\leq \imath_+(s) < q_{\hat{\imath}_+^N}\left( \frac{l}{L} \right) \right\}, \quad l=1\ldots, L
        \end{equation*}
        as well as the mid point of each
        \begin{equation*}
            \imath_l = \frac{L}{N} \sum_{s \in J_l} \imath_+(s)
        \end{equation*}
    \item For each $l$, we generate a random sample $\imath_-^{N,L}$ and $\imath_{spoof}^{N,L}$ of $N/L$ points each drawn from $K(\imath_l, \cdot)$ and $K_{spoof}(\imath_l, \cdot)$, respectively.
    \item For each $l$ we compute the Wasserstein distances
        \begin{equation*}
            W_{2}\left( \imath^{N,l}_-, \hat{\imath}_-^{N,L} \right) \quad \text{and}\quad W_{2}\left( \imath^{N,l}_{spoof}, \hat{\imath}_-^{N,L} \right)
        \end{equation*}
        where $\hat{\imath}^{N,L}_-$ is the discrete distribution out of the sample $\hat{\imath}_-(s)$ for $s$ in $J_l$.
        This is an approximation for the Wasserstein distance
        \begin{equation*}
            W_2\left( \imath_-, \hat{\imath}_n^N \right)|\imath^N_+ \approx \imath_l \quad \text{and} \quad W_2\left( \imath_{spoof}, \hat{\imath}_n^N \right)|\imath^N_+ \approx \imath_l
        \end{equation*}
    \item we aggregate all together and define the indicators
        \begin{align*}
            d\left( \imath_-, \hat{\imath}^N_- | \hat{\imath}^N_+  \right) & := \frac{1}{L}\sum_{l=1}^LW_{2}\left( \imath^{N,l}_-, \hat{\imath}_-^{N,L} \right)\\
            d\left( \imath_{spoof}, \hat{\imath}^N_- | \hat{\imath}^N_+  \right) & := \frac{1}{L}\sum_{l=1}^LW_{2}\left( \imath^{N,l}_{spoof}, \hat{\imath}_-^{N,L} \right)
        \end{align*}
\end{enumerate}
\begin{remark}
    To enhance the accuracy of this indicator, we run step 3 to 5 a couple of times with different samples and average again.
\end{remark}
%

As for the Kernel approximation, we fit $\left(\imath_-, \imath_+\right)$ to a bivariate normal distribution
\begin{equation*}
    \left(\imath_-, \imath_+\right) \sim \mathcal{N} \left(\mu_1,\mu_2,\sigma_1,\sigma_2, \rho\right)
\end{equation*}
thus the conditional distribution is also a normal distribution and $K( y, \cdot)$ is its density function
\begin{equation*}
    \imath_- | \imath_+ = y   \sim \mathcal{N} \left(\mu_1 + \frac{\sigma_1}{\sigma_2} \rho(y - \mu_2), (1 - \rho^2)\sigma_1^2\right)
\end{equation*}

Similarly, we fit $\left(\imath_{spoof}(\imath_-), \imath_+\right)$ to a bivariate skewnormal distribution
\begin{equation*}
    \left(\imath_{spoof}(\imath_-), \imath_+\right) \sim \mathcal{S} \mathcal{N}\left( \alpha, \xi, \Omega\right)
\end{equation*}
where $\alpha = [\alpha_1, \alpha_2]^\top$, $\xi = [\xi_1, \xi_2]^\top$, $\Omega = \begin{bmatrix}
    w_1 & w \\
    w & w_2
\end{bmatrix}$.
It can be derived that
\begin{equation*}
    K_{spoof}(y , \cdot)  = \phi\left(\frac{\cdot - \xi_1^c}{\sqrt{w_{11.2}}}\right) \frac{\Phi\left(\alpha_1 \sqrt{\omega_1} ( \cdot - \xi_1^c) + x_0^\prime\right)}{\Phi(x_0)}
\end{equation*}
where $\phi$, $\Phi$ are the density function and cumulative distribution function of a standard normal distribution, and 
\begin{align*}
&\xi_1^c = \xi_1 + \frac{\omega}{\omega_2}(y - \xi_2 ), \quad
w_{11.2}  = w_1 - \frac{w^2} {w_2}\\
&\bar{\alpha}_2 = \frac{\alpha_2 + \sqrt{\frac{w^2}{w_1 w_2}} \alpha_1 } {\sqrt{ 1 + \frac{w_{11.2} }{w_1}\alpha_1^2}}, \quad
x_0  = \frac{\bar{\alpha}_2}{ \sqrt{w_2}} (y - \xi_2)\\
&x_0^\prime = \sqrt{ 1 + \frac{w_{11.2} }{w_1}\alpha_1^2} x_0
\end{align*}
Figure \ref{fit_joint_dis} shows that $K$ and $K_{spoof}$ fit well with real data in Figure \ref{joint_dis}.
\begin{figure}[H]
    \centering
    \begin{minipage}[c]{0.5\textwidth}
        \centering
        \includegraphics[width=\textwidth,keepaspectratio]{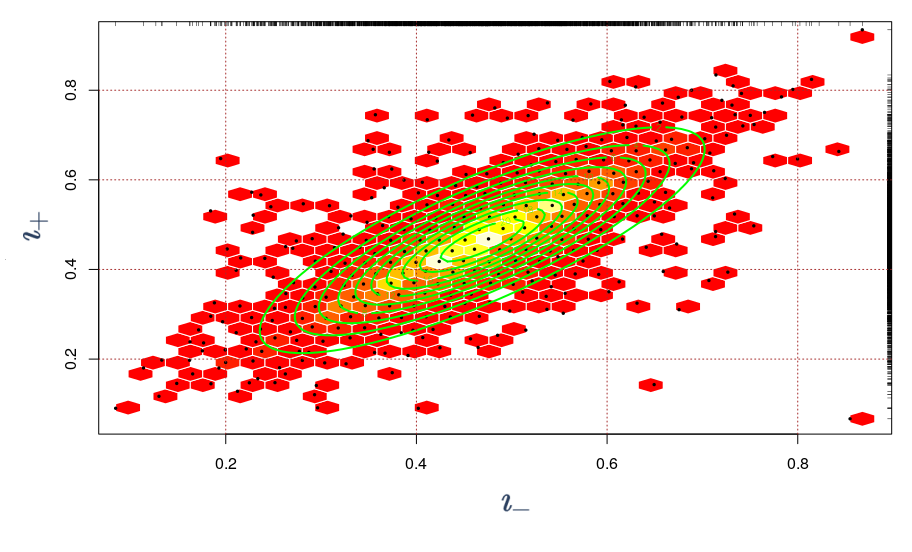}
    \end{minipage}%
    \begin{minipage}[c]{0.5\textwidth}
        \includegraphics[width=\textwidth,keepaspectratio]{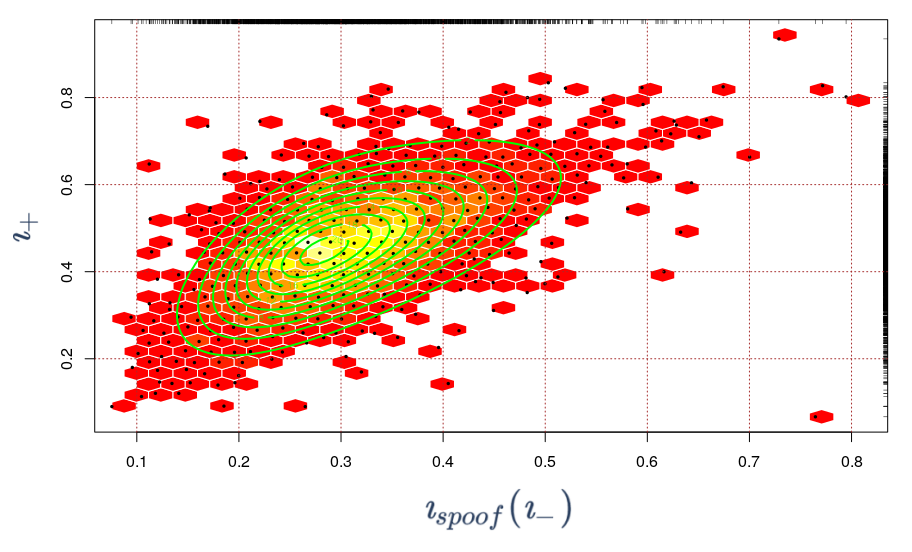}
    \end{minipage}
    \caption{Left panel: Empirical joint distribution of $ \left(\imath_-, \imath_+ \right) $ and $K$ (green contours) . Right panel: Joint distribution of $\left(\imath_{spoof}(\imath_-), \imath_+ \right)$ and $K_{spoof}$ (green contours).}
    \label{fit_joint_dis} 
\end{figure}

\bibliographystyle{abbrvnat}
\bibliography{biblio} 

\end{document}